\documentclass[12pt]{article}

\usepackage[a4paper]{geometry}

\usepackage{latexsym}
\usepackage{amssymb}
\usepackage{euscript}
\usepackage{mathrsfs}
\usepackage{euscript}
\usepackage{xspace}
\usepackage{url}
\usepackage{tabularx}
\usepackage{booktabs}
\usepackage{graphicx}
\usepackage{xcolor}
\usepackage{enumerate}
\usepackage{soul}\sethlcolor{lightgray}

\usepackage{stackengine,scalerel}
\usepackage{color}

\usepackage[sort]{natbib}
\setcitestyle{authoryear,round,citesep={,},aysep={,},yysep={,},notesep={, }}

\usepackage[frenchb,english]{babel}
\usepackage[utf8]{inputenc}
\usepackage[T1]{fontenc}
\usepackage{lipsum}

\usepackage[pagewise]{lineno}

\usepackage[centertags,reqno]{amsmath}
\usepackage{mathtools}
\usepackage[thmmarks,amsmath]{ntheorem}
    \RequirePackage{latexsym}
    \RequirePackage{amssymb}
    \theoremnumbering{arabic}
    \theoremstyle{break}
    \theoremsymbol{}
    \theorembodyfont{\itshape}
    \theoremheaderfont{\normalfont\bfseries}
    \theoremseparator{}
\newtheorem{proposition}{Proposition}

\newtheorem{theorem}[proposition]{Theorem}
\newtheorem{corollary}[proposition]{Corollary}
\newtheorem{lemma}[proposition]{Lemma}
    \theoremheaderfont{\normalfont\bfseries}
    \theorembodyfont{\upshape}
    \theoremsymbol{\ensuremath{\,\,\lrcorner}}
\newtheorem{definition}[proposition]{Definition}
    \theoremsymbol{\ensuremath{\,\,\Diamond}}
    \theoremheaderfont{\normalfont\bfseries}
    \theorembodyfont{\upshape}
\newtheorem{example}[proposition]{Example}
    \theoremheaderfont{\normalfont\bfseries}
    \theoremsymbol{\ensuremath{\ \bullet}}
\newtheorem{remark}[proposition]{Remark}

\theoremstyle{nonumberbreak}
\theoremstyle{nonumberbreak}
\theorembodyfont{\normalfont}
\theoremsymbol{}

\newtheorem{convention}{Convention}
    \theoremheaderfont{\scshape}
    \theoremsymbol{\ensuremath{\ \Box}}
    \newtheorem{proof}{Proof}
    \theoremsymbol{}
    \theoremheaderfont{\normalfont\bfseries}
    
    \theoremsymbol{}
    \theoremheaderfont{\normalfont\bfseries}
    
    \theoremsymbol{\par\hfill Fin Ajout}
    
    \theoremclass{LaTeX}

    \makeatletter
\newtheoremstyle{denisproof}%
  {\item[\rlap{\vbox{\hbox{\hskip\labelsep \theorem@headerfont
          ##1\theorem@separator}\hbox{\strut}}}]}%
  {\item[\rlap{\vbox{\hbox{\hskip\labelsep \theorem@headerfont
          ##1\ ##3\theorem@separator}\hbox{\strut}}}]}
\makeatother
\theoremstyle{denisproof}
    \theorembodyfont{\normalfont}
    \theoremheaderfont{\scshape}
\theoremsymbol{\ensuremath{\ \Box}}

\theoremstyle{break}
    \theoremsymbol{\ensuremath{\ \bullet}}
    \theorembodyfont{\normalfont}
    \theoremheaderfont{\normalfont\bfseries}

\AtBeginDocument{\renewcommand{\iff}{\Leftrightarrow}}

\newcommand{\acceptable}{acceptable\xspace}
\newcommand{\unacceptable}{unacceptable\xspace}
\newcommand{\Acceptable}{Acceptable\xspace}
\newcommand{\Unacceptable}{Unacceptable\xspace}

\newcommand{\txif}{\text{ if }}

\newcommand{\Not}[1]{\text{\textit{Not}}\lbrack #1 \rbrack}

\newcommand{\card}[1]{\left| #1 \right|}
\newcommand{\ens}[1]{\lbrace #1 \rbrace}

\newcommand{\quotient}[2]{\mathord{#1}/\mathord{#2}}
\renewcommand{\emptyset}{\varnothing}   
\newcommand{\vide}{\varnothing}

\newcommand{\Max}{\operatorname{Max}}
\newcommand{\Min}{\operatorname{Min}}

\newcommand{\Natp}{\mathbb{N}^{+}}

\newcommand{\wrt}{w.r.t.\ }
\newcommand{\resp}{resp.\ }
\newcommand{\ie}{i.e.,\xspace}
\newcommand{\eg}{e.g.,\xspace}
\newcommand{\wl}{w.l.o.g.\xspace}

\newcommand{\Elec}{\textsc{Electre}\xspace}
\newcommand{\TRI}{\textsc{Electre Tri}\xspace}
\newcommand{\TRIs}{\textsc{ETri}\xspace}

\newcommand{\lTB}{\textsc{Electre Tri-B}\xspace}
\newcommand{\lnTB}{\textsc{Electre Tri}-nB\xspace}

\newcommand{\nTB}{\TRIs-nB\xspace}
\newcommand{\nTBI}{\TRIs-nB-I\xspace}
\newcommand{\TB}{\TRIs-B\xspace}
\newcommand{\TBI}{\TRIs-B-I\xspace}
\newcommand{\nTBpc}{\nTB-pc\xspace}
\newcommand{\nTBIpc}{\nTBI-pc\xspace}
\newcommand{\nTBpd}{\nTB-pd\xspace}
\newcommand{\nTBIpd}{\nTBI-pd\xspace}

\newcommand{\TC}{\textsc{ETri-C}\xspace}
\newcommand{\nTC}{\textsc{ETri}-nC\xspace}

\newcommand{\TBpc}{\mbox{\TB-pc}\xspace}
\newcommand{\TBpd}{\mbox{\TB-pd}\xspace}

\newcommand{\lineari}{$i$-\linear}

\newcommand{\linear}{linear\xspace}

\newcommand{\rel}{\mathrel{T}}
\newcommand{\irel}{\mathrel{\rel^{\iota}}}
\newcommand{\arel}{\mathrel{\rel^{\alpha}}}
\newcommand{\jrel}{\mathrel{\rel^{\sigma}}}

\newcommand{\brel}{\mathrel{V}}

\newcommand{\M}{\ensuremath{(\Ms)}}
\newcommand{\Ms}{\ensuremath{E}}

\newcommand{\MMM}{\ensuremath{(\widetilde{E})}}

\newcommand{\D}{\ensuremath{(\D1)}}

\newcommand{\rS}{\mathrel{S}}

\newcommand{\rI}{\mathrel{I}}
\newcommand{\rP}{\mathrel{P}}

\newcommand{\rU}{\mathrel{U}}

\newcommand{\rSi}{\mathrel{S_{i}}}

\newcommand{\rIi}{\mathrel{I_{i}}}
\newcommand{\rPi}{\mathrel{P_{i}}}
\newcommand{\rVi}{\mathrel{V_{i}}}
\newcommand{\rUi}{\mathrel{U_{i}}}
\newcommand{\rWi}{\mathrel{W_{i}}}

\newcommand{\Surc}{\mathrel{S}}
\newcommand{\pSurc}{\mathrel{P}}
\newcommand{\aSurc}{\mathrel{P}}

\newcommand{\Set}{X}
\newcommand{\Seti}{\ensuremath{\Set_{i}}}
\newcommand{\N}{\ensuremath{N}}

\newcommand{\Aa}{\ensuremath{\mathcal{A}}}
\newcommand{\Uu}{\ensuremath{\mathcal{U}}}
\newcommand{\PART}{\ensuremath{\langle\Aa,\Uu\rangle}}

\newcommand{\As}{\ensuremath{\Aa_{*}}}

\newcommand{\F}{\ensuremath{\mathcal{F}}}  
\newcommand{\Fs}{\ensuremath{\mathcal{F}_{*}}}

\newcommand{\Pro}{\mathcal{P}}

\newcommand{\relationnormale}{\succsim}
\newcommand{\arelationnormale}{\succ}
\newcommand{\srelationnormale}{\sim}

\newcommand{\qodom}{\relationnormale}
\newcommand{\aqodom}{\arelationnormale}

\newcommand{\essi}{\srelationnormale}

\newcommand{\relsi}{\relationnormale_{i}}
\newcommand{\esi}{\srelationnormale_{i}}
\newcommand{\asi}{\arelationnormale_{i}}

\makeatletter
\let\@fnsymbol\@alph
\makeatother

\title{A theoretical look at \lnTB and related sorting models%
\,\thanks{Authors are listed alphabetically. They have contributed equally.}}
\author{Denis Bouyssou\,\thanks{LAMSADE, UMR\,7243, CNRS,
Universit\'{e} Paris-Dauphine, PSL Research University, 75016 Paris, France,
e-mail: \protect\url{bouyssou@lamsade.dauphine.fr}. \textbf{Corresponding author}.}
\and
Thierry Marchant\,\thanks{Ghent University, Department of Data Analysis,
H.\ Dunantlaan, 1, B-9000 Gent, Belgium,
e-mail: \protect\url{thierry.marchant@UGent.be}.}
\and
Marc Pirlot\,\thanks{Universit\'{e} de Mons, rue de Houdain 9, 7000 Mons, Belgium,
e-mail: \protect\url{marc.pirlot@umons.ac.be}.}}

\date{June 30, 2021}

\begin{document}
\maketitle

\begin{abstract}
\TRI is a set of methods designed to sort
alternatives evaluated on several criteria
into ordered categories. In these methods, alternatives are assigned to categories by comparing them with reference profiles that represent either the boundary or central elements of the category.  The original \lTB method  uses one limiting profile for separating a category from the category below.
A more recent method, \lnTB, allows one to use several limiting profiles
for the same purpose.
We investigate the properties of \lnTB using a conjoint measurement framework.
When the number of limiting profiles used to define each category is not restricted,
\lnTB is easy to characterize axiomatically and is found to be equivalent
to several other methods proposed in the literature.
We extend this result in various directions.

\smallskip

\end{abstract}

\section{Introduction}\label{se:introduction}
\TRI
\footnote{We often abbreviate \TRI as \TRIs in what follows.}
 is a family of methods for sorting alternatives evaluated on several criteria into ordered categories. The principle of these methods is that they assign an alternative to a category by comparing it with profiles specifying levels on each criterion. Comparisons are made by using an \emph{outranking} relation which is typical of the \textsc{Electre} methods. In its original version, \TB
\citep{YuWei92,RoyBouyssou93}, each profile represents the limit between a category and the category below. Therefore, they are called \emph{limiting profiles}. In contrast, in \TC \citep{AlmeidaDiasFigueiraRoy2010}, each category is represented by a typical profile, therefore called \emph{central profile}.

For an introduction to the \textsc{Electre} methods, we refer the reader to \citet[][Ch.~8]{BeltonStewart01}. Overviews of these methods can be found in
\citet[Ch.\ 5 \& 6]{RoyBouyssou93}, \citet{FigueiraELECTREHandbook2010},
\citet{FigueiraELECTREJMCDA2013}, and \citet{FigueiraMousseauRoy2016}.

Recently, \citet{FernandezEJOR2017} proposed a method called
%
\lnTB. It is an extension of \TB, and, thus, uses limiting profiles.
Whereas \TB uses one limiting profile per category,
\nTB allows one to use several limiting profiles
for each category.

\nTB deserves close attention for at least two reasons.
First, as explained in \citet{BouyssouMarchantEJOR2015},
\TRIs can be considered as a real success story within the \textsc{Electre} family of methods.
A closely related model, the NonCompensatory Sorting (NCS) model, has received a fairly complete axiomatic analysis in
\citet{BouyssouMarchant2007:I,BouyssouMarchant2007:II}.
\TRIs has been applied to a large variety of real world problems
\citetext{see the references in \citealp[][Sect.~6]{AlmeidaDiasFigueiraRoy2010}, as well as
\citealp[][Ch.\ 6, 10, 12, 13, 15, 16]{BisdorffEtAl2015}}.
Many techniques have been proposed for the elicitation of the parameters of this method
\citep[see the references in][Sect.~1]{BouyssouMarchantEJOR2015}.

Second, the extension presented with \nTB is most welcome.
Since outranking relations are not necessarily complete,
one may easily argue that it is 
natural to try to characterize
a category using several limiting profiles, instead of just one.
Moreover, compared to \TB, \nTB gives more flexibility to the decision-maker
to define categories using limiting profiles, as observed by \citet[][Remark~3, p.~217]{FernandezEJOR2017}\,\footnote{Let
us also mention that \citet{FernandezEJOR2017} is the last
paper on \textsc{Electre} methods published by Bernard Roy, the founding father of \textsc{Electre} methods,
before he passed away at the end of 2017.
}.

In this paper, we analyze \nTB from a theoretical point of view.
Our aim is to give a complete characterization of this method without any supplementary hypotheses. This is, in a sense, in contrast with \citet{BouyssouMarchant2007:I,BouyssouMarchant2007:II} who characterize a model close to \TB, which is not exactly \TB \citep[it differs from it, in particular, by  considering ``quasi-criteria'' instead of the more general ``pseudo-criteria'' used in \TB, see][pp.~55--56, for definitions]{RoyBouyssou93}. As far as we know, this is the first time that an axiomatic foundation is provided for a complete outranking method (encompassing the construction of the outranking relation and the exploitation phase). The usefulness of such axiomatic analyses
has been discussed elsewhere and will not be repeated here
\citep{BouyssouPirlot2015AOR,DekelLipman2010,GilboaEtAlAxiomaticsTHEO2019}.
Our main finding is that, if the number of profiles used to delimit each category
is not restricted, the axiomatic analysis of \nTB is easy and rests
on a condition, linearity, that is familiar in the analysis of sorting models
\citep{Goldstein91JMP,BouyssouMarchant2007:I,BouyssouMarchant2007:II,BouyssouMarchant2010EJORadditive,GMS04EJOR,
GMS02Control,GrecoMatarazzoSlowinski01Colorni}.
Our simple result shows the equivalence
between \nTB and many other sorting models proposed in the literature.
It could also allow one to use elicitation or learning techniques developed for these other models for the application of \nTB.
This is useful since \citet{FernandezEJOR2017} did not propose any elicitation technique
\citep[an elicitation technique was suggested afterwards in][]{FernandezSoftComputing2019}.

The rest of this text is organized as follows.
In the next section, we recall the definitions of \TB and \nTB. We motivate the theoretical investigation that follows by analyzing an example of an \nTB model.
Section~\ref{se:Notation} introduces our notation and framework.
Section~\ref{sec:model:m} presents our main results about the pseudo-conjunctive version of \nTB.
Section~\ref{sec:model:ext} presents various extensions of these results.
A final section discusses our findings.
An appendix, containing supplementary material to this paper, will allow us to keep
the text of manageable length. Its content will be detailed when needed.

\section{\nTB: definitions and examples}\label{se:motivation}
For the ease of future reference, we first recall the definitions of \TB and \nTB
. For keeping it simple, we limit ourselves to sorting alternatives into two categories, say the ``\acceptable'' and the ``\unacceptable''.
For a more detailed description,
we refer the reader to
\citet{YuWei92,RoyBouyssou93,MousseauSlowinskiZielniewicz2000,FernandezEJOR2017}.
We refer to \citet{BouyssouMarchantEJOR2015} for an analysis of the importance of the various \TRI methods
within the set of all \textsc{Electre} methods.

In the second subsection, we informally analyze an example of an \nTB model in order to motivate the theoretical investigation conducted in the rest of the paper.

\subsection{\TB and \nTB}\label{ssec:etrib_nB}
All \TRI methods are based on the definition of an outranking relation. There are several ways of defining such a relation.
\subsubsection{The outranking relations in \textsc{Electre}~III and in  \textsc{Electre}~I}\label{sssec:outranking}
A crisp outranking relation $S$ (with asymmetric part $P$) comparing pairs of alternatives as in \textsc{Electre}~III
\citep[see][p.~284--289]{RoyBouyssou93} is built by cutting a valued relation $\sigma$ at a certain level $\lambda$. The value associated to each pair in the relation $\sigma$ is called the \emph{outranking credibility index}. It implements (see formula \eqref{eq:degreeCredib} below) the principle of outranking, \ie an alternative $x$ outranks an alternative $y$ if $x$ is at least as good as $y$ on a sufficiently important set of criteria (concordance) and $x$ is  unacceptably worse than $y$ on no criterion (non-discordance). Let $x, y$ be two alternatives respectively represented by their evaluations $(g_1(x), \ldots g_i(x), \ldots g_n(x))$, $(g_1(y), \ldots, g_i(y), \ldots, g_n(y))$ \wrt $n$ criteria. For all $i=1, \ldots, n$, $g_i$ is a real valued function defined on the set of alternatives. 

The concordance index $c(x,y) = \sum_{i=1}^{n} w_i c_i(g_i(x), g_i(y))$, where $w_i \geq 0$ is the importance weight of criterion $i$ (we assume \wl that weights sum up to 1) and $c_i(g_i(x), g_i(y))$ is a function represented in Figure \ref{fig:shapesConcDiscEIII}. Its definition involves the determination of $qt_i$ (resp. $pt_i$), the \emph{indifference} (resp. \emph{preference}) threshold.

\begin{figure}[h!!!]
  \centering
  \includegraphics[width=10cm]{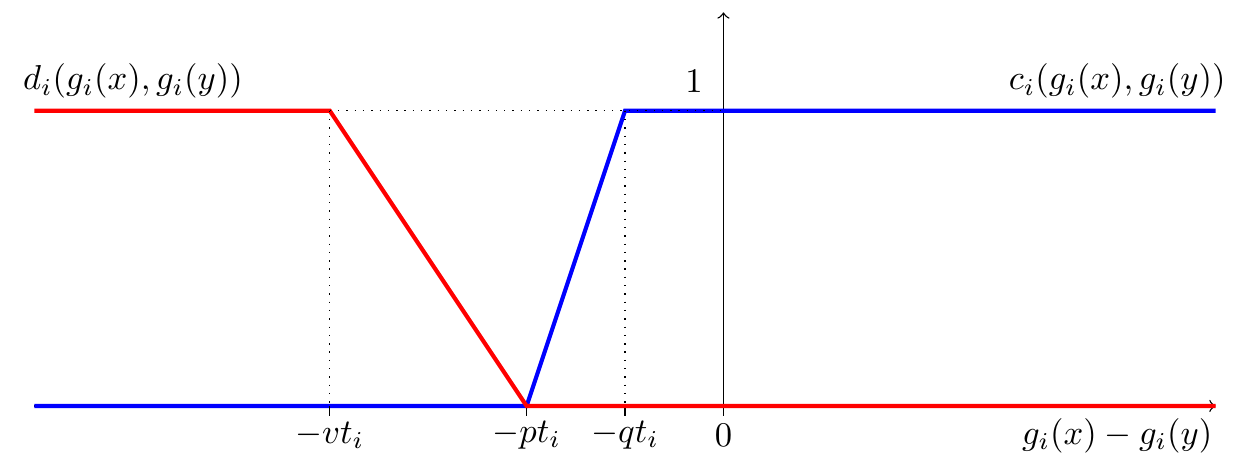}
  \caption{Shapes of the single criterion concordance index $c_i(g_i(x), g_i(y))$ and the discordance index $d_i(g_i(x), g_i(y))$ in \textsc{Electre} III}\label{fig:shapesConcDiscEIII}
\end{figure}

The discordance index $d_i(g_i(x), g_i(y))$, also represented in Figure \ref{fig:shapesConcDiscEIII}, uses an additional parameter $vt_i$, the \emph{veto} threshold\footnote{For the sake of simplicity, the thresholds $qt_i, pt_i$ and $vt_i$ are taken as constant. Nothing in the sequel depends on this option. They could be considered as variable provided appropriate conditions are enforced, actually ensuring that the corresponding weak preference, preference and veto relations form an homogenous chain of semiorders \citep[see][p.~56 and pp.~140--141 for details]{RoyBouyssou93}.}.

The outranking credibility index $\sigma(x,y)$ is computed as follows:
\begin{equation}\label{eq:degreeCredib}
  \sigma(x,y)= c(x,y) \prod_{i: d_i(g_i(x) ,g_i(y)) > c(x,y)} \frac{1-d_i(g_i(x) ,g_i(y))}{1-c(x,y)}.
\end{equation}
Alternative $x$ outranks alternative $y$, \ie $xSy$, if $\sigma(x,y) \geq \lambda$, with $.5 \leq \lambda \leq 1$.

In order that $x$ outranks $y$, $c(x,y)$ has to be greater than or equal to $\lambda$. This index is ``locally compensatory'' in the sense that, for each $i$, there is an interval (namely, $[-pt_i, -qt_i]$) for the differences $g_i(x) - g_i(y)$ on which the single criterion concordance index increases linearly and these indices are aggregated using a weighted sum. Discordance also is gradual in a certain zone (namely $[-vt_i, -pt_i]$); it comes into play only when the discordance index $d_i(g_i(x) ,g_i(y))$ is greater than the overall concordance index $c(x,y)$.

A simpler, more ordinal, version of the construction of an outranking relation stands in the spirit of \textsc{Electre}~I. It is also more amenable to theoretical investigation: see the characterization of outranking relations \citep{BouyssouPirlot2016Conjoint} and the analysis of the noncompensatory sorting model \citep{BouyssouMarchant2007:I,BouyssouMarchant2007:II}.
It differs from the above mainly by the shapes of the single criterion concordance and discordance indices (see Figure \ref{fig:shapesConcDiscEI}).

\begin{figure}[h!!!]
  \centering
  \includegraphics[width=10cm]{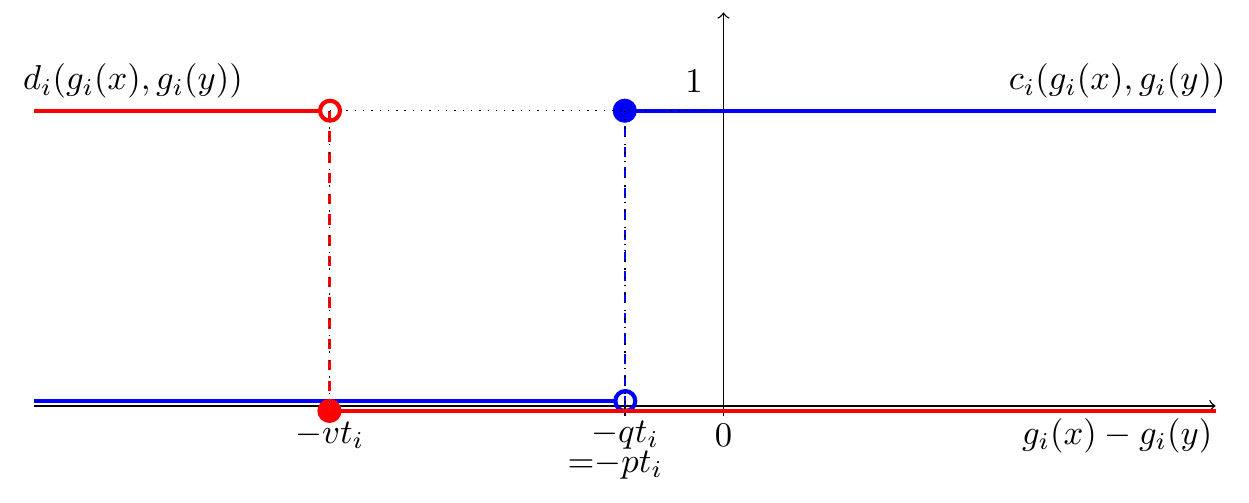}
  \caption{Shapes of the single criterion concordance index $c_i(g_i(x) ,g_i(y))$ and the discordance index $d_i(g_i(x) ,g_i(y))$ in the style of \textsc{Electre}~I. Empty (resp. filled) circles indicate included (resp. excluded) values.} \label{fig:shapesConcDiscEI}
\end{figure}

The preference and indifference thresholds are confounded, which implies that there is no linear ``compensatory'' part in $c_i(g_i(x) ,g_i(y))$; discordance only occurs in an all-or-nothing manner. The overall concordance index $c(x,y) = \sum_{i=1}^{n} w_i c_i(g_i(x) ,g_i(y))$, as above. In this construction, $x$ outranks $y$, \ie $xSy$, if $\sigma(x,y) \geq \lambda$, with

\begin{equation}\label{eq:degreeCredibOrdinal}
  \sigma(x,y)= c(x,y) \prod_{i=1}^{n} (1-d_i(g_i(x) ,g_i(y))),
\end{equation}
\ie $xSy$ if $c(x,y) \geq \lambda$ and $d_i(g_i(x) ,g_i(y)) =0$, for all $i$. Note that
\begin{equation*}
  c(x,y) = \sum_{i:g_i(x) \geq g_i(y)-qt_i} w_i.
\end{equation*}
We thus have $c(x,y) \geq \lambda$ if the sum of the weights of the criteria on which $x$ is indifferent or strictly preferred to $y$ is at least equal to $\lambda$. Subsets of criteria of which the sum of the weights is at least $\lambda$ will be called \emph{winning coalitions} (of criteria).

\smallskip
Notice that both the \textsc{Electre~III} outranking relation defined by means of \eqref{eq:degreeCredib} and the \textsc{Electre~I} outranking relation defined by means of \eqref{eq:degreeCredibOrdinal} respect the \emph{dominance relation}\footnote{The (weak) dominance relation $\geq$ is a reflexive and transitive relation on the set of alternatives, that is defined as follows: $x \geq y$ if $g_i(x) \geq_i g_i(y)$, for all $i$. This is the relation denoted $\Delta_F$ by \citet[][p.~61]{RoyBouyssou93}, $F$ referring to a family of criteria. \label{foot:dominance}} $\geq$. This is easily seen by observing that both formulae \eqref{eq:degreeCredib} and \eqref{eq:degreeCredibOrdinal} are nondecreasing in $g_i(x)$ and nonincreasing in $g_i(y)$, for all $i$. We note this fact in the following proposition for further reference.

\begin{proposition}\label{prop:SrespectsDom}
Let $S$ denote an outranking relation of \textsc{Electre~III} or \textsc{Electre~I} type. The relation $S$ respects the dominance relation $\geq$, \ie for all alternatives $x,y,z,w$,
\begin{equation*}
  [xSy, z\geq x \ \textrm{and} \ y \geq w] \quad \Rightarrow \quad zSw.
\end{equation*}
\end{proposition}

\subsubsection{\TB}\label{sssec:etriB}

The sorting of an alternative $x$ into category \Aa\ (\acceptable) or \Uu\ (\unacceptable) is based upon the comparison of
$x$ with a limiting profile\,%
%
%
$p$ using the relation $\Surc$.

In the pessimistic version of \TB, now known, following \citet{AlmeidaDiasFigueiraRoy2010},
as the pseudo-conjunctive version (\TBpc), we have, for all $x \in \Set$,
\begin{equation*}
x \in \Aa \Leftrightarrow x \Surc p.
\end{equation*}
In the optimistic version of \textsc{Electre Tri}, now known as the pseudo-disjunctive
version (\TBpd), we have, for all $x \in \Set$,
\begin{equation*}
x \in \Aa \Leftrightarrow \Not{p \pSurc x},
\end{equation*}
where $\pSurc$ is the asymmetric part of $\Surc$. Consequently, we have $x \in \Uu \iff p \pSurc  x$.



\subsubsection{\nTB}\label{sssec:etrinB}

We now have a set of $k$ limiting profiles $\Pro = \{p^{1}, p^{2}, \dots, p^{k}\}$.
This set of limiting profiles must be such that, for all $p, q \in \Pro$,
we have $\Not{p \aSurc q}$.

In the pseudo-conjunctive version of \nTB (\nTBpc, for short), we have that
\begin{equation*}
x \in \Aa \Leftrightarrow
\begin{cases}
x \Surc p        & \text{ for some } p \in \Pro, \text{ and}\\
\Not{q \aSurc x} & \text{ for all }  q \in \Pro,
\end{cases}
\end{equation*}
and $x \in \Uu$, otherwise.

In the pseudo-disjunctive version of \nTB (\nTBpd, for short), we have that
\begin{equation*}
x \in \Uu \Leftrightarrow
\begin{cases}
p \aSurc x        & \text{ for some } p \in \Pro, \text{ and}\\
\Not{x \aSurc q}  & \text{ for all }  q \in \Pro,
\end{cases}
\end{equation*}
and $x \in \Aa$, otherwise.

\TBpc and \TBpd are particular cases of \nTBpc and \nTBpd, respectively. In this section, we consider only \nTBpc and omit the suffix ``pc''. We shall only turn back, briefly, to \nTBpd in Section~\ref{sse:disjunctive}. Following \citet{FernandezEJOR2017}, unless otherwise mentioned, we use the \textsc{Electre}~III outranking relation $S$ defined by means of \eqref{eq:degreeCredib}. The version of \nTB using the \textsc{Electre}~I outranking relation $S$ defined via \eqref{eq:degreeCredibOrdinal} will be referred to as \nTBI.

\begin{remark}\label{rem:ETRI1et3monotone}
Using Proposition \ref{prop:SrespectsDom}, it is easy to see that \nTB and \nTBI \emph{respect the dominance relation} $\geq$, \ie are monotone \wrt this relation. In particular,  if $y$ dominates  the \acceptable alternative $x$, then $y$ is \acceptable. Symmetrically, if $x$ is \unacceptable and dominates $y$, then $y$ is unacceptable.
\end{remark}

\subsubsection{The family of \TRI methods}\label{ssec:familyETRI}
To ease the reading, we summarize the different variants of the \TRIs methods considered in the sequel as well as their interrelationships. The variants of \nTB appear in Table \ref{ta:variantsTBn} on the left. Their version using only one limiting profile, \ie the different variants of \TB, appear on the right of the same table.

\begin{table}[h!!!]
  \centering
  \begin{tabular}{llcll}
    \hline
    (a) & \nTB with veto & \mbox{ } \quad \mbox{ } & (e)& \TB with veto  \\
    (b) & \nTB without veto & & (f)& \TB without veto\\
    (c) & \nTBI with veto & & (g) & \TBI with veto\\
    (d) & \nTBI without veto  & & (h) & \TBI without veto. \\
    \hline
  \end{tabular}
  \caption{Variants of \nTB and \TB}\label{ta:variantsTBn}
\end{table}

Model (a) contains model (b), which contains model (d). Model (c) (resp. (d)) differs from model (a) (resp. (b)) in that it uses the outranking relation in \Elec~I instead of \Elec~III.

The relationships between the \TB models are the same as between the homologous \nTB models. 
The NonCompensatory Sorting (NCS) model analyzed by \citet{BouyssouMarchant2007:I,BouyssouMarchant2007:II} generalizes (h), while the NCS model with veto generalizes (g). In both cases, the generalization lies in that the winning coalitions of criteria are not necessarily determined by means of additive weights. Model (h) is called the Majority Rule sorting model (MR-Sort) in the literature \citep{LeroyMousseauPirlot2011,Sobrie2018}.

Note that another type of \TRIs methods has been proposed, namely \TC \citep{AlmeidaDiasFigueiraRoy2010} and \nTC \citep{AlmeidaDiasFigueiraRoy2012}. These are based on a different logic (using \emph{central} profiles instead of \emph{limiting} profiles) as analyzed in \citet{BouyssouMarchantEJOR2015}. We shall not consider them in this paper for lack of place.

\subsection{An example of an \nTB model}\label{sse:exampleEtrinB}
Consider alternatives evaluated on three criteria. Each criterion value belongs to the $[0, 10]$ interval. In practice, such evaluations have a limited precision. Let us assume that evaluations are integers or half-integers.

Assume that the alternatives are partitioned into two classes \Aa\ and \Uu\ by an \nTB model with 2 limiting profiles. Let these profiles be $p^1 = (8,7,5)$ and $p^2=(5,6,8)$. The indifference, preference and veto thresholds are, respectively, $qt_i=1$, $pt_i=2$ and $vt_i=4$, the same for all criteria $i=1,2,3$. All criteria have the same weight $w_i=\frac{1}{3}$ and the cutting threshold $\lambda = .6$.

It is readily verified that ${p^1 S p^2}$  and ${p^2 S p^1}$ so that none of the profiles strictly outranks the other.

\subsubsection{Minimally \acceptable alternatives}\label{ssse:precision1half}
Let us apply the above model. The set of possible evaluations for criterion $i$ is  
$X_i = \{0, .5, 1, 1.5, \ldots, 9.5, 10\}$, for $i=1,2,3$, and the set of all possible alternatives is $X=\prod_{i=1}^{3} X_i$. 
Each alternative is thus represented by an evaluation vector: for any $x \in X$, $x=(x_1, x_2, x_3)$, with $x_i = g_i(x)$, $i=1, 2, 3$. Since each alternative is identified with its evaluation vector, the dominance relation $\geq$ on $X$ is asymmetric. Therefore, it is a partial order.

Since \nTB is monotone \wrt the dominance relation $\geq$, which is a partial order on $X$, and since there are finitely many alternatives, the set \Aa\ of \acceptable alternatives has a finite number of minimal elements \As\ that we shall call \emph{minimally \acceptable alternatives} (see Section~\ref{ssec:minInPosets} for further justification). The set \Aa\ is 
the set of alternatives that dominate at least one alternative in \As. It contains $\As$. Decreasing the performance of a minimally \acceptable alternative by any amount on any criterion produces an \unacceptable alternative.

Let us determine the set \As. We first focus on $p^1$. Given the granularity of the evaluations, for satisfying $c(x,p^1)\geq .6$, the index $c_i(x_i, p^1_i)$, which takes only the values 0, .5 and 1,
\begin{itemize}
  \item must be 1 for two criteria $i \in \{1,2,3\}$; it can be 0 for the third one,
  \item or must be 1 for one criterion and take the value .5 on the other two.
\end{itemize}
Consider the alternatives of the form $x=(7,6,x_3)$. For them, $c(x, p^1) \geq \frac{2}{3} > .6$. We have $d_3(x_3, p^1_3) = .75$ if $x_3 = 1.5$ and $d_3(x_3, p^1_3) = .5$ if $x_3 = 2$. Therefore
$\sigma(7,6,1.5) = \frac{2}{3}\times \frac{1/4}{1/3}= .5 < .6$  and  $\sigma(7,6,2) = \frac{2}{3}>.6$ since $d_3(2, p^1_3) = .5 <  c(x,p^1)$. Therefore, $(7,6,2)$ is minimal in $\Aa$ and, by a similar reasoning, we have that $(7,4,4)$ and $(5,6,4)$ are also minimal.

Consider now the second type of minimal alternatives. For example, for $x=(7, 5.5, 3.5)$, we have $\sigma(x,p^1) = c(x,p^1)= 1 \times 1/3 + 1/2 \times 1/3 + 1/2 \times 1/3 = 2/3 > 0.6$. Clearly, none of the performances of $x$ can be decreased by .5 without resulting in an \unacceptable alternative. Therefore, $(7, 5.5, 3.5)$ is minimal and, by a similar reasoning, we see that $(6.5, 6, 3.5)$ and $(6.5,5.5,4)$ are minimal too.

Applying the same analysis to the second profile $p^2$, yields the complete description of the set $\As$ of minimally \acceptable elements displayed in Table \ref{ta:minAccAltEx} (the first (resp. second) row corresponds to profile $p^1$ (resp. $p^2))$.

\begin{table}[h!!!]
  \centering
  \begin{tabular}{cccccc}
  $(7,6,2)$ & $(7,4,4)$ & $(5,6,4)$ & $(7, 5.5, 3.5)$ & $(6.5, 6, 3.5)$ & $(6.5,5.5,4)$ \\
  $(4,5,5)$ & $(4,3,7)$ & $(2,5,7)$ & $(4,4.5,6.5)$ & $(3.5,5,6.5)$ & $(3.5,4.5,7)$. \\
\end{tabular}

  \caption{List of minimally \acceptable alternatives in case evaluations are integers or half-integers}\label{ta:minAccAltEx}
\end{table}

None of these 12 alternatives dominates another. The number of elements in \As\ is thus 12.

\begin{remark}\label{rem:simplerPheno}
Let us briefly discuss the consequences of using a similar \nTBI model, using the \textsc{Electre}~I outranking relation, instead of the more classical version above. We keep the same two limiting profiles $p^1, p^2$ and the same parameters except for $qt_i$ and $pt_i$ that we both set equal to 1 and $vt_i$ that we set to 3. It is easy to see that there are three minimally \acceptable alternatives \wrt $p^1$ which are $(7,6,2)$, $(7,4,4)$ and $(5, 6,4)$. The minimally \acceptable alternatives \wrt $p^2$ are $(4,5,5)$, $(4, 3, 7)$ and $(2, 5, 7)$. The number of minimally \acceptable alternatives is half the one in Table \ref{ta:minAccAltEx}. The minimally \acceptable alternatives in this simplified model are identical to the first three ones in each row of Table \ref{ta:minAccAltEx}. The last three ones in each row are not ``represented'' in the simplified model. They correspond to alternatives for which the distinction between thresholds $pt_i$ and $qt_i$ plays an important role.
\end{remark}


\subsubsection{Observations}\label{sssec:observations}
We emphasize the following observations supported by the above example.
\begin{enumerate}
  \item From the analysis of the above example, it results that an alternative is assigned to \Aa\ by the \nTB model iff it is equal or dominates one of the twelve alternatives listed in Table \ref{ta:minAccAltEx}. Therefore, this model is equivalent to another \nTB model with different parameters. The latter has the 12 alternatives $\Pro' = \{{p'}^{1}, \ldots,  {p'}^{j}, \dots, {p'}^{12}\}$ listed in Table~\ref{ta:minAccAltEx} as limiting profiles. For all $i=1,2,3$, $w_i'=1/3$, $pt_i' = qt_i'=0$ and $vt_i'$ is a large number, \eg $vt_i'=10$. We set $\lambda' = 1$. With this model, $c'(x,{p'}^j) \geq \lambda' =1 $ iff $x_i \geq {p'}_i^{j}$ for all $i$. There is no veto effect since $d_i(x_i, {p'}_i^j) \leq c'(x,{p'}^j)$ whenever the condition $c'(x,{p'}^j) \geq \lambda'$ is fulfilled and whatever the value of $vt'_i$. We call such a model an \emph{unanimous} \nTB model in the sequel.
   \item While the scale $X_i$ of each criterion $i=1, 2, 3$ has 21 levels (all integers and half integers between 0 and 10), only 6 of them are distinguished by appearing as distinct values of the $i$th coordinate in the 12 minimally \acceptable alternatives listed in Table~\ref{ta:minAccAltEx}. The \nTB model distinguishes only the 7 classes of equivalent evaluations that are delimited by these 6 values. For instance, on the scale of criterion $i=1$, the 6 values that make a difference are $7, 6.5, 5, 3, 2.5, 1$. They are the different values taken by the first coordinate of the alternatives in Table~\ref{ta:minAccAltEx}. This means that the model's assignments to \Aa\ or \Uu\ induce a weak order $\succsim_1$ on $X_1$ that is coarser than the natural order on the set of integers and half-integers in $[0,1]$. This weak order $\succsim_1$ (with its asymmetric part denoted $\succ_1$ and its symmetric part $\sim_1$) on $X_1$ is as follows:
       \begin{gather*}
         [10 \sim_1 9.5 \sim_1 9 \sim_1 8.5 \sim_1 8 \sim_1 7.5 \sim_1 7] \succ_1 6.5 \succ_1 [6 \sim_1 5.5 \sim_1 5] \\
        \succ_1 [4.5 \sim_1 4 \sim_1 3.5 \sim_1 3] \succ_1 2.5 \succ_1 [ 2 \sim_1.5  \sim_1 1] \succ_1 [0.5 \sim_1 0].
       \end{gather*}
       This implies, for example, the following. If the evaluation of $x$ on the first criterion is 6, decreasing it to 5 does not change the assignment of the alternative. Such a weak order with 7 equivalence classes is defined on each criterion by the model.
  \item In the process of aiding a decision maker (DM) to make a decision by \emph{eliciting} her preference in a question-and-answer session, \nTB may be a useful tool because the principle of concordance/non-discordance at the root of the method is intuitively appealing. The perspective is different when the parameters of the method are not elicited through actual interaction with a DM but have to be \emph{learned} on the basis of an (often limited) number of assignment examples. The minimal number of examples that allows us to determine a sorting model is important whenever learning the model is the issue. Assume that an oracle tells you that \nTB is the model used by the DM for sorting the alternatives into two categories \Aa\ and \Uu. The oracle gives you the values of all the model's parameters including the number and the definition of limiting profiles. What is the minimal number of assignment questions you have to ask the DM just to verify that the oracle is not cheating on you? The most efficient questioning strategy is asking the decision maker to assign all minimally \acceptable alternatives (that can be determined according to the model indicated by the oracle). If the DM assigns them all to \Aa, then it is still necessary to ask her to assign all maximally \unacceptable alternatives. If the DM assigns them all to \Uu, then the oracle's model is the right one. So, in particular, the number of minimally \acceptable alternatives (\ie the limiting profiles of the unanimous \nTB equivalent model) is important in a learning context. From this point of view, \nTB appears as rather complex since the set of minimally \acceptable alternatives it induces tends to be large. In the above example with 3 criteria, 2 limiting profiles and criteria scales composed of integers and half-integers, this number is 12. It grows rapidly, for instance, with the criteria scales precision. If we apply the same model to the case the criteria scales are rationals with one decimal digit ranging in [0,10] (\ie 101 levels on each criterion scale instead of 21), the number of minimally \acceptable alternatives grows up to 192 (see Supplementary material, Appendix~\ref{app:sec:ExampleDecimal}). Therefore, in a learning perspective, the question of approximating an \nTB model by a simpler one, \ie a model determining relatively few minimally \acceptable alternatives is important.
\end{enumerate}

\subsection{Goal of the paper}\label{ssec:goal}

In Sections \ref{se:Notation} and \ref{sec:model:m}, we analyze, in a conjoint measurement framework \citep[][Ch.~6 and~7]{KrantzLuceSuppesTversky71Vol1}, an assignment model, Model $\M$, that is closely related  to the \nTBI method presented above. Just as \nTB generalizes \TB, Model $\M$ generalizes the noncompensatory sorting model studied by \citet{BouyssouMarchant2007:I,BouyssouMarchant2007:II} to the case in which several limiting profiles are used to sort the alternatives.

We place ourselves in a conjoint measurement framework because it is the usual one in decision theory and it has been used in previous works analyzing the \textsc{Electre} methods. Analyzing sorting methods in this framework means that any alternative in a Cartesian product can be sorted into categories and that an a priori linear ordering of each criterion scale is not postulated. A weak order on each criterion scale, if it exists, will be revealed by the partition. Working in such a framework does not restrict the generality of the study. Indeed, in case each criterion scale is linearly ordered and a partition respects the dominance relation determined by these orders, then the partition does reveal a weak order on each scale, possibly coarser than the a priori linear orders, but compatible with them. This was illustrated in item~2 of Section~\ref{sssec:observations}.

Our main finding is that, if the number of limiting profiles is not bounded above, the axiomatic analysis of Model $\M$ is easy and rests
on a condition, linearity, that is familiar in the analysis of sorting models
\citep{Goldstein91JMP,BouyssouMarchant2007:I,BouyssouMarchant2007:II,BouyssouMarchant2010EJORadditive,GMS04EJOR,
GMS02Control,GrecoMatarazzoSlowinski01Colorni}.
Our simple result shows the equivalence
between Model $\M$ and several other sorting models, in particular, the unanimous model introduced above in the example.

We prove, in Section \ref{ssec:EvsEtrinBpc}, that the \nTB model, which uses the \textsc{Electre}~III outranking relation, is \emph{equivalent} to the \nTBI model and to the Decomposable model \eqref{def:model:D} (that will be defined below, in Section~\ref{subsec:Goldstein}). By ``equivalent'', we mean that, for all particular \nTB model, there is an \nTBI model that determines the same partition. The parameters of these equivalent models are possibly different. In particular, we emphasize that the sets of limiting profiles used in equivalent models usually differ, also in cardinality.
Our theoretical analysis gives insight into the issue of learning such models on the basis of assignment examples (see Section \ref{ssec:learning}).

\section{Notation and framework}\label{se:Notation}

Although the analyses presented in this paper can easily be extended to cover
the case of several ordered categories, we will mostly limit ourselves
to the study of the case of \emph{two} ordered categories.
This will allow us to keep things simple, while
giving us a sufficiently rich framework to present our main points.

Similarly, we suppose throughout that the set of objects
to be sorted
is \emph{finite}. This is hardly a limitation with applications of sorting methods in mind.
The extension to the general case is not difficult but
calls for developments that would obscure
our main messages\,%
\footnote{In fact our framework allows us to deal with some infinite sets of objects:
all that is really required is that the set of equivalence classes of each set
$\Set_{i}$ under the equivalence $\esi$ is finite, see below.}.

\subsection{The setting}
Let $n\geq 2$ be an integer and
$\Set = \Set_{1} \times \Set_{2} \times \cdots \times \Set_{n}$ be
a \emph{finite} set of objects\footnote{Note that, in contrast with Section~\ref{se:motivation}, the sets $X_i$ are not necessarily sets of real numbers. They also need not be the range of a function $g_i$ evaluating the alternatives \wrt criterion $i$. The set $X_i$ can be any finite set, not necessarily ordered \emph{a priori}. \label{foot:Xi}}.
Elements $x, y, z,\ldots$ of $\Set$ will be interpreted
as alternatives evaluated on a set
$\N = \ens{1, 2, \ldots, n}$ of attributes\,%
\footnote{We use a standard vocabulary for binary relations. For the convenience of the reader,
all terms that are used in the main text are defined in Appendix~\ref{app:binary},
given as supplementary material. See also,
\eg \citet{AleskerovBouyssouMonjardet07Book,DoignonMonjardetRoubensVincke88JMP,PirlotVincke92,RoubensVincke85Book}.}.
%
Any element $x \in X$ is thus an $n$-dimensional vector $x=(x_1, \ldots, x_i, \ldots, x_n)$, with $x_i \in X_i$, for all $i \in N$. For all $x, y \in X$ and $i \in N$, we denote by $(x_i, y_{-i})$ the element $w \in X$ such that $w_i = x_i$ and, for all $j \neq i$, $w_j = y_j$. In other words, $w=(x_i, y_{-i})$ is obtained by replacing the $i$th component of $y$, \ie $y_i$, by $x_i$.


Our primitives consist in a twofold \emph{partition} $\PART$ of the set $\Set$. This means that
the sets $\Aa$ and $\Uu$ are nonempty and disjoint and that their union is the entire
set $\Set$. Our central aim is to study various models allowing to represent the information
contained in $\PART$.
We interpret the partition $\PART$ as the result of
a sorting method applied to the alternatives in $\Set$.
Although the ordering of the categories is not part of our primitives,
it is useful to interpret the set $\Aa$ as containing \Acceptable  objects,
while $\Uu$ contains \Unacceptable ones.

We say that an attribute $i \in \N$ is influential for $\PART$ if
there are $x_{i}, y_{i}\in X_{i}$ and $a_{-i}\in \Set_{-i}$ such that
$(x_{i}, a_{-i}) \in \Aa$ and $(y_{i}, a_{-i}) \in \Uu$.
We say that an attribute is degenerate if it is not influential.
Note that the fact that $\PART$ is a partition implies that there is at least one influential
attribute in $\N$. A degenerate attribute has no influence whatsoever
on the sorting of the alternatives and may be suppressed from $\N$.
Hence, we suppose henceforth that all attributes are influential for $\PART$.


A twofold partition $\PART$ induces on each $i \in \N$ a binary relation defined letting, for all $i \in \N$ and all
$x_{i}, y_{i} \in X_{i}$,
\begin{equation}\label{eq:sim_i}
x_{i} \esi y_{i} \txif
\big[\forall a_{-i} \in X_{-i}, (y_{i},a_{-i}) \in \Aa \iff (x_{i},a_{-i}) \in \Aa\big].
\end{equation}
This relation is always reflexive, symmetric and transitive, \ie is an equivalence.
We omit the simple proof of the following \citep[see][Lemma~1, p.~220]{BouyssouMarchant2007:I}.
\begin{lemma}\label{lemma1:Gold}
For all $x, y \in \Set$ and all $i \in \N$,
\begin{enumerate}
\item\label{it1:lemma1:Gold}
$[y \in \Aa$ and $x_{i}\esi y_{i}]$ $\Rightarrow$ $(x_{i}, y_{-i}) \in \Aa$,
\item\label{it2:lemma1:Gold}
$[x_{j} \essi_{j} y_{j}$, for all $j \in N]$ $\Rightarrow$ $[x \in \Aa \Leftrightarrow y \in \Aa]$.
\end{enumerate}
\end{lemma}
This lemma will be used to justify the convention made later
in Section~\ref{sub:model:m}.


\subsection{A general measurement framework}\label{subsec:Goldstein}
\citet{Goldstein91JMP} suggested the use of conjoint measurement techniques
for the analysis of twofold and threefold partitions of a set of multi-attributed alternatives.
His analysis was  rediscovered and  developed
in \citet{GrecoMatarazzoSlowinski01Colorni} and \citet{GMS02Control}.
We briefly recall here the main points of the analysis in the above papers
for the case of twofold partitions. We follow 
\citet{BouyssouMarchant2007:I}.

Let $\PART$ be a partition of $\Set$.
Consider a measurement model, henceforth the \emph{Decomposable} model, in which, for all $x \in \Set$,
\begin{equation}\label{def:model:D}\tag{$D1$}
x \in \Aa \Leftrightarrow F(u_{1}(x_{1}), u_{2}(x_{2}), \ldots, u_{n}(x_{n})) > 0,
\end{equation}
where $u_{i}$ is a real-valued function on $\Seti$ and
$F$ is a real-valued function on $\prod_{i=1}^{n}u_{i}(\Seti)$
that is \emph{nondecreasing}
in each argument\,%
\footnote{In Model~\eqref{def:model:D}, notice that we could have chosen
to replace the strict inequality by a nonstrict one.
The two versions of the model are equivalent, as shown in
\citet[Rem.~8, p.~222]{BouyssouMarchant2007:I}. The same is true for Model $(D2)$.}.
The special case of Model \eqref{def:model:D} in which
$F$ is supposed to be \emph{increasing} in each argument, is called Model $(D2)$.
%
%
Model $(D2)$ contains as a particular case the additive model
for sorting in which, for all $x \in \Set$,
\begin{equation}\label{def:model:Add}\tag{$Add$}
x \in \Aa \Leftrightarrow \sum_{i=1}^{n}u_{i}(x_{i}) > 0,
\end{equation}
that is at the heart of the UTADIS technique
\citep{Jacquet95UTADIS} 
and its variants \citep{Zopounidis00MHDIS,Zopounidis00PREFDIS,GrecoMousseauSlowinski2010UTADIS}.
It is easy to check\,%
\footnote{When $\Set$ is finite, it is clear that the variant of
Model~\eqref{def:model:Add} in which the strict inequality is replaced
by a nonstrict one is equivalent to
Model~\eqref{def:model:Add}.}
that there are twofold partitions that can be obtained in Model
$(D2)$ but that cannot be obtained in Model \eqref{def:model:Add} (see Supplementary material, Appendix~\ref{app:AddNotEqDecomp}).

In order to analyze Model \eqref{def:model:D},
we define on each $\Seti$ the binary relation $\relsi$ letting,
for all $x_{i}, y_{i}\in \Seti$,
\begin{equation}\label{eq:trace}
x_{i} \relsi y_{i} \txif [\mbox{for all } a_{-i} \in \Set_{-i},
(y_{i}, a_{-i}) \in \Aa \Rightarrow (x_{i}, a_{-i}) \in \Aa].
\end{equation}
It is not difficult to see that, by construction,
$\relsi$ is reflexive and transitive.
We denote by $\asi$ (resp.\ $\esi$) the asymmetric (resp.\ symmetric) part of $\relsi$
(hence, the relation $\esi$ coincides with the one defined by \eqref{eq:sim_i}).

We say that the partition $\PART$ is \linear on attribute $i \in \N$ (condition \lineari)
if, for all $x_{i}, y_{i} \in \Seti$ and all $a_{-i}, b_{-i}\in \Set_{-i}$,
\begin{equation}\label{axiome:Linear:Gold}\tag{\lineari}
\left.
\begin{array}{c}
(x_{i}, a_{-i}) \in \Aa\\
\mbox{and}\\
(y_{i}, b_{-i}) \in \Aa
\end{array}
\right\}
\Rightarrow
\left\{
\begin{array}{c}
(y_{i}, a_{-i}) \in \Aa,\\
\mbox{or}\\
(x_{i}, b_{-i}) \in \Aa.
\end{array}
\right.
\end{equation}
The partition is said to be \linear if it is \lineari, for all $i \in \N$.
This condition was first proposed in \citet{Goldstein91JMP}, under the name ``context-independence'',   and generalized
in \citet{GrecoMatarazzoSlowinski01Colorni} and \citet{GMS02Control} (these authors call it ``cancellation property'').
The adaptation of this condition to the study of binary relations, adaptation
first suggested by
\citet{Goldstein91JMP}, is central in the analysis of the nontransitive
decomposable models
presented in \citet{BouyssouPirlot99Meskens,BouyssouPirlot02NTDCM,BouyssouPirlot02Intra}.

The following lemma takes note of the consequences of condition \lineari
on the relation $\relsi$ and shows that linearity is necessary for Model \eqref{def:model:D}.
Its proof can be found in \citet[Lemma~5, p.~221]{BouyssouMarchant2007:I}.
\begin{lemma}\label{lemma:rels:complete:Gold}
\begin{enumerate}
\item\label{it1:lemma:rels:complete:Gold}
Condition \lineari holds iff $\relsi$ is complete,
\item\label{it2:lemma:rels:complete:Gold}
If a partition $\PART$ has a representation in Model \eqref{def:model:D} then it is \linear.
\end{enumerate}
\end{lemma}
%

The following proposition is due to \citet[Theorem 2]{Goldstein91JMP} and
\citet[Theorem 2.1, Part 2]{GrecoMatarazzoSlowinski01Colorni}.

\begin{proposition}\label{prop:Goldstein:monotone}
Let $\PART$ be a twofold partition of a set $\Set$. Then:
\begin{enumerate}[(i)]
\item there is a representation of $\PART$
in Model \eqref{def:model:D} iff it is \linear,
\item if $\PART$ has a representation in Model \eqref{def:model:D},
it has a representation in which, for all
$i \in \N$, $u_{i}$ is a numerical representation of $\relsi$,
\item moreover, Models \eqref{def:model:D} and $(D2)$ are equivalent.
\end{enumerate}
\end{proposition}

\begin{proof}
See, \eg \citet[Proposition~6, p.~222]{BouyssouMarchant2007:I}
%
\end{proof}

\subsection{Partitions respecting a dominance relation}\label{ssec:partRespDom}
Footnote~\ref{foot:Xi} has emphasized that it is not necessarily the case that $X_i$ is a subset of the reals and the range of a function $g_i$ evaluating the alternatives \wrt criterion $i$, for all $i \in N$.
In the case a partition respects the dominance relation determined by pre-existing linear orderings of the criteria scales (see Remark~\ref{rem:ETRI1et3monotone}, for a definition), we note the following result.

\begin{proposition}\label{prop:partRespectsDom}
Let $X = \prod_{i=1}^{n} X_i$, where the finite set $X_i$ is endowed with a linear order $\geq_i$, for all $i$. Let $\PART$ be a twofold partition of $X$ which respects the dominance relation $\geq$ determined by the linear orders $\geq_i$. Then $\PART$ is linear and the weak order $\succsim_i$ induced by the partition is compatible with the linear order $\geq_i$, for all $i$, \ie for all $x_i, y_i \in X_i$, $x_i \geq_i
y_i$ entails $x_i \succsim_i y_i$.
\end{proposition}

\begin{proof}
If $x_i \geq_i y_i$, condition \eqref{eq:trace} is fulfilled, since the partition respects dominance, and therefore $x_i \succsim_i y_i$. Since $\geq_i$ is complete, so is $\succsim_i$, for all $i$. Therefore $\PART$ is linear (Lemma \ref{lemma:rels:complete:Gold}.1).
\end{proof}

The fact that, in general, $\succsim_i$ does not distinguish (\ie considers as equivalent) some pairs that are strictly ordered by $\geq_i$ is illustrated in Section \ref{sssec:observations}, item~2.

We noticed in Remark~\ref{rem:ETRI1et3monotone}, that \nTBpc and \nTBIpc respect the dominance relation. Therefore, we have the following corollary of Proposition~\ref{prop:partRespectsDom}.
\begin{corollary}\label{cor:ETRI1et3linear}
The twofold partitions determined by \nTBpc and \nTBIpc are linear.
\end{corollary}

\subsection{Interpretations of the Decomposable model \eqref{def:model:D}}\label{ssubsec:D1:int}
The framework offered by the Decomposable model \eqref{def:model:D} is quite flexible.
It contains many other sorting models
as particular cases. 
We already observed that it contains Model \eqref{def:model:Add} as a particular case.
\citet{BouyssouMarchant2007:I} have reviewed various possible interpretations of
Model \eqref{def:model:D}. They have shown that
both the pseudo-conjunctive and the pseudo-disjunctive
variants  of \TBI (see Table~\ref{ta:variantsTBn}) enter into this framework. In particular, they have characterized, within the Decomposable model, the NCS model, which is a generalization (without additive weights) of \TBI (pseudo-conjunctive).

\citet[Theorem 2.1, Parts 3 and 4]{GrecoMatarazzoSlowinski01Colorni}
\citep[see also][Theorem 2.1]{GMS02Control}
have proposed two equivalent
reformulations of the Decomposable model \eqref{def:model:D}.
The first one  uses ``at least'' decision rules. The second one
uses a binary relation to compare alternatives to a profile.
We refer to \citet{BouyssouMarchant2007:I} and to the original papers for details.

\section{Main Results}\label{sec:model:m}
\subsection{Definitions}\label{sub:model:m}
The following definition synthesizes the main features of \nTBIpc, the version of \nTBpc using the \textsc{Electre}~I outranking relation (see  Section~\ref{ssec:etrib_nB}).
The main differences \wrt \nTBIpc are that:
\emph{(i)} we do not suppose that the real-valued functions $g_{i}$ are given beforehand and
\emph{(ii)} we do not use additive weights combined with a threshold to determine the
winning coalitions.
Actually, the model defined below is a \emph{multi-profile version of the noncompensatory sorting} model with veto analyzed in \citet{BouyssouMarchant2007:I}, exactly in the same way as \nTB is a multi-profile version of \TB.  For notational simplicity, we shall refer to it as Model $\M$ (``\emph{E}'', for \TRI) in the sequel.  

\begin{definition}\label{def:model:M}
We say that a partition
$\PART$ has a representation in Model $\M$ if:
\begin{itemize}
\item for all $i \in \N$, there is a semiorder
$\rSi$ on $X_{i}$ (with asymmetric part
$\rPi$ and symmetric part $\rIi$),

\item for all $i \in \N$, there is a strict semiorder $\rVi$
on $X_{i}$ that is included in $\rPi$ and is the asymmetric part of a semiorder
$\rUi$,

\item $(\rSi, \rUi)$ is a \emph{homogeneous} nested chain of semiorders
and ${\rWi} = {\rS_{i}^{wo}} \cap {\rU_{i}^{wo}}$ is a weak order
that is compatible with both
$\rSi$ and $\rUi$,

\item there is a set of subsets of attributes  $\F \subseteq 2^{N}$ such that,
for all $I, J \in 2^{N}$, $[I \in \F$ and $I \subseteq J]$
$\Rightarrow J \in \F$,

\item there is a binary relation $\rS$ on $X$ (with
symmetric part $\rI$ and asymmetric part $\rP$) defined
by
\begin{equation*}
x \rS y \iff \left[S(x,y) \in \F \text{ and } V(y,x) = \vide\right],
\end{equation*}

\item there is a set $\Pro = \{p^{1}, \ldots , p^{k}\} \subseteq X$ of $k$
limiting profiles, 
such that for all $p, q \in \Pro$, $\Not{p \rP q}$,
\end{itemize}

such that

\begin{equation}\label{eq:M}\tag{\Ms}
x \in \Aa  \iff
\begin{cases}
x \rS p    &\text{ for some } p \in \Pro \quad \text{and}\\
\Not{q \rP x}  &\text{ for all } q \in \Pro,
\end{cases}
\end{equation}
where
\begin{equation*}
S(x, y) = \{i \in \N : x_{i} \rSi y_{i}\},
\end{equation*}
and
\begin{equation*}
V(x, y) = \{i \in \N : x_{i} \rVi y_{i}\}.
\end{equation*}

We then say that $\langle (\rSi, \rVi)_{i\in\N}, \F, \Pro\rangle$
is a representation of $\PART$ in Model $\M$.
Model $(\Ms^{c})$ is the particular case of Model $\M$,
in which there is a representation that
shows no discordance effects, \ie
in which all relations $\rVi$ are empty.
Model $(\Ms^{u})$ is the particular case of Model $\M$, in which there is a representation
that requires unanimity, \ie such that $\F = \{N\}$.
\end{definition}

\begin{table}[hbt]
\centering
\begin{tabular}{lll}
\toprule
$\M$        & $\langle (\rSi, \rVi)_{i\in\N}, \F, \Pro\rangle$             & General model  \\
$(\Ms^{c})$   & $\langle (\rSi, \vide)_{i\in\N}, \F, \Pro\rangle$          & Based on concordance \\
$(\Ms^{u})$   & $\langle (\rSi, \vide)_{i\in\N}, \F = \{\N\}, \Pro\rangle$ & Based on unanimity \\
\bottomrule
\end{tabular}
\caption{Model $\M$ and its variants.\label{table:model:m}}
\end{table}

Table~\ref{table:model:m} summarizes the models defined above.
It should be clear that $\M$ is closely related to \nTBIpc (see Remark \ref{rem:simpleretrinB} below).
It does not use
criteria but uses attributes, as is traditional in conjoint measurement.
Moreover, it does not use an additive weighting
scheme combined with a threshold to determine winning coalitions but
uses instead a general family $\F$ of subsets of attributes
that is compatible with inclusion
\citep[see also][]{BouyssouMarchant2007:I}. Note that, when the set of limiting profiles $\Pro$ is restricted to be a singleton, Model $\M$ is \emph{exactly} the noncompensatory sorting model (NCS) studied by \citet{BouyssouMarchant2007:I}.

\begin{remark}\label{rem:simpleretrinB}
Any partition determined by an \nTBIpc model has a representation in Model $\M$. We illustrate this fact using the example of the \nTBIpc model described in Remark~\ref{rem:simplerPheno}. Note that this way of constructing a representation in Model $\M$ is applicable to any \nTBIpc model.

For all $i=1, 2, 3$, $X_i$ is the set of integers and half-integers between 0 and 10. The semiorder $S_i$ on $X_i$ is defined using the threshold $pt_i = qt_i =1$, \ie for all $x_i, y_i \in X_i$, $x_i S_i y_i$ iff $x_i \geq y_i - 1$. We have $x_i S_i y_i$ iff $c_i(x_i,y_i) =1$. Similarly, the strict semiorder $V_i$ is defined using the threshold $vt_i =3$ by $x_i V_i y_i$ iff $x_i > y_i + 3$. We have $x_i V_i y_i$ iff $d_i(y_i, x_i) = 1$. The subsets of attributes in $\F$ are all sets of two or three attributes since $c(x,y) \geq .6$ if and only if, for at least two criteria, $x_i S_i y_i$ and, therefore, $|S(x,y)| \geq 2$. The pair $(x,y)$ belongs to the outranking relation $S$ iff $|S(x,y)| = |\{i: x_i S_i y_i\}| \geq 2$ and $V(y,x) = \{i: y_i V_i x_i\} = \emptyset$, \ie it is never the case that $y_i \geq x_i + 3.5$. With these definitions of $S_i, V_i$ and $\F$, the acceptable alternatives in the example are exactly these which satisfy \M.

It is easy to see that the model in the example is equivalent to a model based on unanimity, \ie a model ($\Ms^{u}$), using as limiting profiles the three first alternatives in each row of Table~\ref{ta:minAccAltEx} and the natural order $\geq$ as the relation $S_i$ on $X_i$.
%
\end{remark}

\begin{remark}
It is clear that Model $(\Ms^{u})$ is a particular case
of Model $(\Ms^{c})$: if unanimity is required
to have $x \rS y$, the veto relations $\rVi$
play no role and can always be taken to be empty.
\end{remark}

The following lemma takes note of elementary consequences of the
fact that $(\rSi, \rUi)$ is a \emph{homogeneous} nested chain of semiorders (we remind the reader that the necessary definitions are recalled in Appendix~\ref{app:binary}, as supplementary material).

\begin{lemma}\label{lem:monotonicity}
\label{lemma:SAndSi}
Let $\langle \Aa, \Uu \rangle$ be a  twofold partition of $X$.
If $\langle \Aa, \Uu \rangle$  is representable in $\M$ 
then,
for all $a = (a_{i}, a_{-i})$, $b=(b_{i}, b_{-i}) \in X$,
all $i \in N$ and all $c_{i} \in X_{i}$,
\begin{subequations}\label{eq:mon:total}
\begin{align}
a \rS b \text{ and } b_{i} \rWi c_{i}  &\Rightarrow a \rS (c_{i}, b_{-i}),\label{eq:mon}\\
a  \rP b \text{ and } b_{i} \rWi c_{i} &\Rightarrow a \rP (c_{i}, b_{-i}),\label{eq:mon:st}\\
a \rS b \text{ and } c_{i} \rWi a_{i}  &\Rightarrow (c_{i},a_{-i}) \rS b,\label{eq:monL}\\
a \rP b \text{ and } c_{i} \rWi a_{i}  &\Rightarrow (c_{i},a_{-i}) \rP b,\label{eq:monL:st}
\end{align}
\end{subequations}
where $\rWi$ denotes a weak order that is compatible with the homogeneous
nested chain of semiorders $(\rSi, \rUi)$.
\end{lemma}
\begin{proof}
%
Let $a' = (c_{i},a_{-i})$ and $b' = (c_{i}, b_{-i})$.
Let us show that \eqref{eq:mon} holds.
Suppose that $a \rS b$, so that $S(a, b)  \in \F$ and
$V(b,a) = \vide$.
Because $b_{i} \rWi c_{i}$, we know
that $S(a, b') \supseteq S(a, b)$. Hence, we have
$S(a, b')  \in \F$.
Similarly, we know that $V(b,a) = \vide$, so that
$\Not{b_{i} \rVi a_{i}}$.
It is therefore impossible that
$c_{i} \rVi a_{i}$ since $b_{i} \rWi c_{i}$ would imply
$b_{i} \rVi a_{i}$, a contradiction.
Hence, $V(b', a) = \vide$ and we have
$a \rS b'$.

Let us show that \eqref{eq:mon:st} holds.
Because $a \rP b$ implies $a \rS b$,
we know from \eqref{eq:mon} that $a \rS b'$.
Suppose now that $b' \rS a$ so that
$S(b', a) \in \F$ and $V(a, b') = \vide$.
Because $b_{i} \rWi c_{i}$, $c_{i} \rSi a_{i}$ implies $b_{i} \rSi a_{i}$,
so that
$S(b, a) \supseteq S(b', a)$, implying
$S(b, a)  \in \F$.
Similarly, we know that $V(a, b') = \vide$, so that
$\Not{a_{i} \rVi c_{i}}$.
It is therefore impossible that
$a_{i} \rVi b_{i}$, since $b_{i} \rWi c_{i}$ would imply
$a_{i} \rVi c_{i}$, a contradiction.
Hence, we must have $V(a, b) = \vide$,
 so that we have
$b \rS a$, a contradiction.

The proof of \eqref{eq:monL} and \eqref{eq:monL:st} is similar.
\end{proof}

The next lemma shows that Model $\M$ implies 
linearity.

\begin{lemma}\label{lem:linearity}
Let $\PART$ be a twofold partition of $\Set = \prod_{i=1}^{n} \Seti$.
If $\PART$ has a representation in Model $\M$ then it is \linear.
\end{lemma}
\begin{proof}
Suppose that we have
$(x_{i}, a_{-i}) \in \Aa$,
$(y_{i}, b_{-i}) \in \Aa$.
Defining the relations $W_i$ as in Lemma~\ref{lemma:SAndSi}, we have either  $x_{i} \rWi y_{i}$ or $y_{i} \rWi x_{i}$.
Suppose that $x_{i} \rWi y_{i}$.
Because $(y_{i}, b_{-i}) \in \Aa$, we know that
$(y_{i}, b_{-i}) \Surc p$, for some  $p \in \Pro$,
and
$\Not{q \aSurc (y_{i}, b_{-i})}$  for all $q \in \Pro$,
Lemma~\ref{lem:monotonicity} implies that
$(x_{i}, b_{-i}) \Surc p$ and
$\Not{q \aSurc (x_{i}, b_{-i})}$  for all $q \in \Pro$.
Hence, $(x_{i}, b_{-i}) \in \Aa$.
The case $y_{i} \rWi x_{i}$ is similar: we start with
$(x_{i}, a_{-i}) \in \Aa$ to conclude that $(y_{i}, a_{-i}) \in \Aa$.
Hence, linearity holds.
%
%
%
%
%
%
%
\end{proof}

In view of Lemma~\ref{lem:linearity}, we therefore know from
Lemma~\ref{lemma:rels:complete:Gold}
that in 
Model $\M$ 
there is, on each attribute $i \in \N$, a weak order
$\relsi$ on $\Seti$ that is compatible with the partition $\PART$.

\begin{convention}
For the analysis of $\PART$ on $\Set = \prod_{i=1}^{n} \Seti$, it is not useful to keep in $X_{i}$ elements
that are equivalent \wrt the equivalence relation $\esi$.
Indeed, if
$x_{i} \esi y_{i}$ then $(x_{i}, a_{-i}) \in \Aa$ iff $(y_{i}, a_{-i}) \in \Aa$
(see Lemma~\ref{lemma1:Gold}).

In order to simplify the analysis, it is not restrictive to suppose that we work with
$\quotient{\Seti}{\esi}$ (\ie the set of equivalence classes in $\Seti$ for the equivalence $\esi$) instead of $X_{i}$ and, thus, on
$\prod_{i=1}^{n} [\quotient{\Seti}{\esi}]$
instead of
$\prod_{i=1}^{n} \Seti$.
This amounts to supposing that
the equivalence $\esi$ becomes the identity relation.
\emph{We systematically make this hypothesis below.} This is \wl
since the properties of a partition on $\prod_{i=1}^{n} [\quotient{\Seti}{\esi}]$
can immediately be extended to a partition on $\prod_{i=1}^{n} \Seti$ (see Lemma~\ref{lemma1:Gold})
and is done for convenience only.
In order to simplify notation, we suppose below
that we are dealing with partitions
on $\prod_{i=1}^{n} \Seti$ for which all relations $\esi$ are trivial.
Our convention implies that each relation $\relsi$ is antisymmetric, so that
the sets  $\Seti$ are linearly ordered by $\relsi$.
\end{convention}

Let us define the relation $\qodom$ on $\Set$ letting, for all $x,y \in \Set$,
\begin{equation*}
x \qodom y \iff x_{i} \relsi y_{i}, \text{ for all } i \in \N.
\end{equation*}
It is clear that the relation $\qodom$ plays the role of a dominance relation
in our conjoint measurement framework. It is a partial order on $\Set$, being reflexive, antisymmetric,
and transitive. This partial order is obtained as a ``direct product of chains''
(the relations $\relsi$ on each $\Seti$)
as defined in \citet[p.~119]{CaspardBookEng2012}.

Before we turn to our main results, it will be useful to take
note of a few elementary observations
about maximal and minimal elements in partially ordered sets (posets), referring
to \citet{DaveyPriestley}, for more details.

\subsection{Minimal and maximal elements in posets}\label{ssec:minInPosets}

Let $\rel$ be a binary relation on a set $Z$.
An element $x \in B \subseteq Z$ is maximal
(\resp minimal)
in $B$ for $\rel$ if there is no $y \in B$ such that $y \arel x$
(\resp $x \arel y$), where $\arel$ denotes the asymmetric part of $\rel$.
The set of all maximal (\resp minimal)
elements in $B \subseteq Z$ for $\rel$ is denoted by $\Max(\rel, B)$
(\resp $\Min(\rel, B)$).

For the record, the following proposition recalls some well-known
facts about maximal and minimal elements of partial orders on finite sets
\citep[][p.~16]{DaveyPriestley}. We sketch its proof in Appendix~\ref{app:maximal} for completeness.

\begin{proposition}\label{prop:minimal}
Let $\rel$ be a partial order (\ie a reflexive, antisymmetric and transitive relation)
on a nonempty set $Z$.
Let $B$ be a finite nonempty subset of $Z$.
Then the set of maximal elements,
$\Max(\rel, B)$, and the set of minimal elements, $\Min(\rel, B)$, in
$B$ for $\rel$ are both nonempty.
For all $x,y \in \Max(\rel, B)$ (\resp $\Min(\rel, B)$)
we have $\Not{x \arel y}$.
Moreover, for all $x \in B$, there is
$y \in \Max(\rel, B)$ and
$z \in \Min(\rel, B)$ such that
$y \rel x$ and $x \rel z$.
\end{proposition}

We will apply the above proposition
to the proper nonempty subset $\Aa$ of the finite set $\Set = \prod_{i=1}^{n} \Seti$,
partially ordered by $\qodom$.

\subsection{A characterization of Model $\M$}\label{sub:results}

We know that $\qodom$ is a partial order on $\Set = \prod_{i=1}^{n} \Seti$.
Because $\PART$ is a twofold partition of $\Set$, we know that
$\Aa \ne \vide$. Because we have supposed $\Set$ to be finite, so is $\Aa$.
Hence, we can apply Proposition~\ref{prop:minimal}
to conclude that the set $\As = \Min(\qodom, \Aa)$ is nonempty.

We are now fully equipped to present our main result.

\begin{theorem}\label{th:main:nB}
Let $\Set = \prod_{i=1}^{n} \Seti$ be a finite set and $\PART$
be a twofold partition of $\Set$.
The partition $\PART$ has a representation in Model $\M$
iff it is \linear. This representation can always be taken to be
$\langle ({\relsi}, {\rVi} = \vide)_{i \in \N}, \F = \{\N\}, \Pro = \As \rangle$.
\end{theorem}
\begin{proof}
We know from Lemma~\ref{lem:linearity} that Model $\M$ implies linearity.
Let us prove the converse implication.
Take, for each $i \in \N$, ${\rSi} = {\relsi}$ and ${\rVi} = \vide$.
Take $\F = \{\N\}$. Hence, we have ${\rS} = {\qodom}$.
Take $\Pro = \As$.
Using Proposition~\ref{prop:minimal}, we know that $\As$ is nonempty and that, for all
$p, q \in \As$, we have $\Not{p \aqodom q}$. Hence, taking $\Pro = \As$
leads to an admissible set of profiles in Model $\M$.

If $x \in \Aa$, we use Proposition~\ref{prop:minimal}
to conclude that there is $y \in \As$ such that
$x \qodom y$, so that we have $x \qodom p$, for some $p \in \Pro$.
Suppose now that, for some $q \in \Pro$, we have $q \aqodom x$.
Using the fact that $\qodom$ is a partial order,
we obtain $q \aqodom p$, contradicting the fact that $p, q \in \As$, in view of
Proposition~\ref{prop:minimal}.
Suppose now that
$x \in \Uu$. Supposing that $x \qodom p$, for some $p \in \Pro = \As$,
would lead to $x \in \Aa$, a contradiction.
This completes the proof. 
\end{proof}

\begin{remark}\label{rem:useless condition}
In the representation in Model $\M$ built in Theorem~\ref{th:main:nB},
the relation $\rS$ is a partial order. When this is so, the condition
stating that $\Not{q \rP x}  \text{ for all } q \in \Pro$ and all $x \in \Aa$, is automatically verified.
Indeed, suppose that, for some $q \in \Pro$ and some $x \in \Aa$, we have
$q \rP x$. Because $x \in \Aa$, there is $p \in \Pro$ such that
$x \rS p$. Transitivity leads to $q \rP p$, violating the condition on the set of profiles.
\end{remark}

\begin{remark}\label{rem:uniqueness}
Under our convention that $\relsi$ is antisymmetric, for all $i \in \N$,
it is clear that, if we are only interested
in representations with $\F = \{\N\}$, the set $\Pro$ must be taken equal to $\As$. Hence, the representation
built above is unique, under our convention about antisymmetry and the constraint that $\F = \{\N\}$.
Without the constraint that $\F = \{\N\}$, uniqueness does not obtain any more, as shown, \eg by Example \ref{ex:base} below.
Since this is not important for our purposes, we do not investigate this point further in this text.
\end{remark}

\subsection{Example}\label{sub:examples}

We illustrate the construction of the representation in Theorem~\ref{th:main:nB} with the example below. 

\begin{example}\label{ex:base}
Let $X = \prod_{i=1}^{3} X_{i}$ with
$X_{1} = X_{2} = X_{3} = \{39, 37, 34, 30, 25\}$.
Hence, $X$ contains $5^{3} = 125$ objects.

Define the twofold partition $\PART$ letting:
\begin{equation*}
(x_{1}, x_{2}, x_{3}) \in \Aa \iff
x_{1} + x_{2} + x_{3} \geq 106.
\end{equation*}

In this twofold partition, the set $\Aa$ contains $32$ objects, while $\Uu$ contains
the remaining $93$ objects.

It is easy to check that all attributes are influential for this partition
and that, on each attribute $i \in \N$, we have
$39 \asi 37 \asi 34 \asi 30 \asi 25$. For instance, for attribute 1, we have:
\begin{equation*}
\begin{aligned}
&(39, 37, 30) \in \Aa &\quad & (37, 37, 30) \notin \Aa,\\
&(37, 39, 30) \in \Aa &\quad & (34, 39, 30) \notin \Aa,\\
&(34, 39, 34) \in \Aa &\quad & (30, 39, 34) \notin \Aa,\\
&(30, 39, 37) \in \Aa &\quad & (25, 39, 37) \notin \Aa.
\end{aligned}
\end{equation*}


This twofold partition has an obvious representation in Model \eqref{def:model:Add}.
Hence it is \linear and also has a representation in Model $\M$.
Considering the representation built in
Theorem~\ref{th:main:nB} with
${\rSi} = {\relsi}$, ${\rVi} = {\vide}$,
$\Pro = \As$ and $\F = \{N\}$,
we obtain\,\footnotemark{} a representation that uses the following 12 profiles:
\footnotetext{We omit details and the
reader is invited to check this example using, \eg his/her favorite spreadsheet software.}
\begin{equation}\label{12prof}
\begin{aligned}
&(37, 37, 34)& &\quad & (39, 34, 34)&  &\quad &(39, 37, 30) &\quad &(37, 30, 39)\\
&(37, 34, 37)& &\quad & (34, 39, 34)&  &\quad &(39, 30, 37) &\quad &(30, 39, 37)\\
&(34, 37, 37)& &\quad & (34, 34, 39)&  &\quad &(37, 39, 30) &\quad &(30, 37, 39).
\end{aligned}
\end{equation}
It is clear that these twelve profiles are pairwise incomparable
\wrt ${\Surc} = {\qodom}$.

For instance, the object $(39, 30, 39)$ belongs to $\Aa$, because
$39+30+39 = 108 \geq 106$.
This object outranks (meaning here, dominates) the two profiles
$(39, 30, 37)$ and $(37, 30, 39)$, but no other.

Notice that this twofold partition is also determined by an \nTBpc model with a single limiting profile $p^1 = (39, 39, 39)$, thresholds $qt_i = 0$, $pt_i = 9$, $vt_i=14$, for $i=1, 2, 3$ and $\lambda = 16/27$.
\end{example}

\section{Remarks and extensions}\label{sec:model:ext}


\subsection{Positioning Model $\M$ \wrt other sorting models}\label{sub:comments}

Theorem~\ref{th:main:nB} gives a simple characterization of Model $\M$.
It makes no restriction on the size of the set of profiles $\Pro$,
except that it is finite.


The proof of Theorem~\ref{th:main:nB}
builds a representation of any partition $\PART$ satisfying linearity in a special case of Model $\M$,
Model $(\Ms^{u})$. This shows that Models $(\Ms^{u})$ and $\M$ are equivalent.
Because, $(\Ms^{u})$ is a particular case of Model $(\Ms^{c})$, this also shows that
Models $\M$, $(\Ms^{u})$ and $(\Ms^{c})$ are equivalent.

In view of Proposition~\ref{prop:Goldstein:monotone},
Model $\M$ is equivalent to Model \eqref{def:model:D} and, hence, to Model $(D2)$. However, Model \eqref{def:model:Add} is not equivalent to Model $\M$. An example of a linear partition that is not representable in Model \eqref{def:model:Add} is given in Appendix \ref{app:AddNotEqDecomp}, as supplementary material. Note also that a characterization of Model \eqref{def:model:Add} in case $X$ is a finite set is known but requires a countably infinite scheme of axioms \citep[see][Appendix B, and especially, Remark 31, p.~32]{BouyssouMarchant2007:additivedecomp}.

Because Model $(D2)$ contains Model \eqref{def:model:Add} as a particular case,
the same is true for
Model $\M$.

We summarize our observations in the following.
\begin{proposition}
\begin{enumerate}
\item Models $\M$, $(\Ms^{c})$, and $(\Ms^{u})$ are equivalent.
\item Models $\M$, \eqref{def:model:D}, and $(D2)$ are equivalent.
\item Model \eqref{def:model:Add}
      is a particular case of  Model $\M$ but not vice versa.
\end{enumerate}
\end{proposition}

The above proposition allows us to position rather precisely
Model $\M$ within the family of all sorting models.

\begin{remark}
As observed in Section \ref{sub:model:m}, when the set $\Pro$ of limiting profiles is restricted to a singleton, Model $\M$ is the noncompensatatory sorting model. The twofold partitions $\PART$ that can be represented in this model have been characterized as being linear and 3v-graded \citep[see][]{BouyssouMarchant2007:I}. The latter property implies that the weak order $\succsim_i$ induced by the partition on $X_i$, for all $i$, distinguishes at most three equivalence classes of evaluations on criterion $i$. In case there is no veto effect, $\succsim_i$ distinguishes at most two equivalence classes on $X_i$ (2-graded). This result characterizes the twofold partitions representable in the noncompensatory sorting model within the set of linear twofold partitions. In other words, it characterizes the particular case of Model $\M$ with one limiting profile within the general Model $\M$.
\end{remark}


%

\subsection{Model $\M$ vs. \nTBpc}\label{ssec:EvsEtrinBpc}

In order to position Model $\M$ \wrt \nTBpc, we assume that, for all $i$, $X_i$ is the range of a real-valued function $g_i$ evaluating the alternatives \wrt criterion $i$.
We know that Model $\M$ is equivalent to Model $(\Ms^{u})$.
These models are characterized by
linearity. But all partitions obtained with the original method \nTBpc satisfy linearity (by Corollary \ref{cor:ETRI1et3linear}).
Therefore, \nTBpc is a particular case of Model $\M$.

Conversely, Model $(\Ms^{u})$ is a particular case of \nTBpc that is obtained taking
the cutting level $\lambda$ to be $1$ and, on all criteria, the preference and indifference thresholds to be equal.
Since unanimity is required, veto thresholds play no role.
Since Model $(\Ms^{u})$ is equivalent to Model $\M$, Model $\M$ is a particular case of \nTBpc.

The same can be said of \nTBIpc, since this model also satisfies linearity (by Corollary~\ref{cor:ETRI1et3linear}) and contains the unanimous model $(\Ms^{u})$. This is also true of the versions without veto of \nTBpc and \nTBIpc. Therefore, $\M$ and all four models in the left part of Table~\ref{ta:variantsTBn} are equivalent models.

We take note of this in the next proposition.

\begin{proposition}\label{prop:EtrinBequivModE}
Models $\M$, \nTBpc (with or without veto) and \nTBIpc (with or without veto) are equivalent models\footnote{We emphasize that, by \emph{equivalent models}, we mean that any partition that has a representation in one of the three models, also has a representation in the two other models, using an appropriate set of parameters. In particular, not only the limiting profiles used in these models are generally different but also the numbers of limiting profiles differ.}.
\end{proposition}

The relationship of $\M$ with \nTBIpc is however tighter than with \nTBpc. Indeed, as shown in Remark~\ref{rem:simpleretrinB}, any \nTBIpc
model has a straightforward representation in Model $\M$ that does not make use of the unanimous model.

\subsection{Unanimous representations}\label{ssec:unanimous}
The reader may be perplexed by the fact that the proof of Theorem~\ref{th:main:nB}
builds a representation in Model $\M$ in which $\F =\{\N\}$.
This is indeed a very particular form of representation.
Notice that there are linear partitions of which the sole representation in Model $\M$ is a unanimous one, \ie with $\F =\{\N\}$. The curious reader will find such an example in Appendix \ref{app:OnlyReprIsUnanim}, as supplementary material.
Hence, representations with $\F =\{\N\}$ are sometimes quite useful and even, may be the only possible ones. We show below that obtaining such a representation from a representation based on concordance, \ie in Model $(\Ms^{c})$, is easy.



Any representation in Model $\M$
can be transformed
into a representation with $\F =\{\N\}$. This is a direct
consequence of Theorem~\ref{th:main:nB}.
When the representation is without discordance,
we show below that the process of building $\Aa$ and then
deriving $\As$ can be avoided.

Suppose we know a representation
$\langle (\rS_{i}, \vide)_{i \in N}, \F, \Pro \rangle$ of
the partition $\langle A, U \rangle$ in $\M$ (in fact, in $(\Ms^{c})$)
and that $\F \neq \{N\}$.

We want to find
a representation such that $\F = \{N\}$. Theorem~\ref{th:main:nB}
ensures that such a representation exists.
It can be built quite efficiently, independently
of the construction used in Theorem~\ref{th:main:nB}.

Let $\Fs = \Min(\supseteq, \F)$, the set of minimal elements in $\F$ \wrt inclusion.
For each $i \in \N$, let $x^{0}_{i}$ be the unique element in $\Seti$ that is minimal for
the linear order $\relsi$.
Define $x^{0}$ accordingly. Moreover, let:
\begin{equation*}
\Pro' = \{ (x^{0}_{-I},p_{I}), \text{ for all } p \in \Pro \text{ and } I \in \Fs \}.
\end{equation*}
It is clear that
$\langle ({\rS_{i}}, \vide)_{i \in N}, \{N\}, \Pro' \rangle$
is a representation of the partition $\langle A, U \rangle$ in $\M$.

\subsection{Variable set of winning coalitions $\F$}

A rather natural generalization of Model $\M$,  called $\MMM$,
is as follows. Instead of considering a single family of winning coalitions
$\F$ that is used to build the relation $\Surc$ and compare
each alternative in $\Set$ to all profiles in $\Pro$, we could use
a family $\F^{p}$ that would be specific to each profile $p \in \Pro$,
with a relation $\Surc^{p}$
that now depends on the profile.

The analysis of Model $\MMM$
is easy. It is simple to check that Model $\MMM$
implies linearity. Indeed, for each relation $\Surc^{p}$, Lemma~\ref{lem:monotonicity}
holds and, hence, linearity cannot be violated. This shows that
Model $\MMM$ is a particular case of Model $\M$ and, hence, is equivalent to it.

\subsection{More than two categories}

Our analysis of Model $\M$ can easily be extended to cover the case of an arbitrary number
of categories. Because this would require the introduction of a rather cumbersome framework,
without adding much to the analysis of two categories, we do not formalize this point. We briefly
indicate how this can be done, leaving the details to the interested readers.

Linearity has been generalized to cope with more than two categories. This is
done in \citet{GMS02Control} and \citet{BouyssouMarchant2007:II}.
The intuition behind this generalization is simple. It guarantees that there is
a weak order on each attribute that is compatible with the ordered partition.


When $\Set$ is finite, this condition is
necessary and sufficient to characterize the obvious generalization of
Model \eqref{def:model:D} that uses more than one threshold, instead of the single
threshold $0$ \citep[see, \eg][Prop.~7, p.~250]{BouyssouMarchant2007:II}.
Moreover, it is easy to check that this condition is satisfied by the natural
generalization of Model $\M$ that uses more than two categories (this involves working with
of a set of profiles $\Pro$ for each of the induced twofold partitions).

Now, the technique used in the proof of Theorem~\ref{th:main:nB}
easily allows one to define a family of profiles for each of the induced twofold
partitions.
It just remains to check that these families of profiles satisfy the constraints
put forward in \citet[Condition~1, p.~216]{FernandezEJOR2017}. This is immediate.

\subsection{Pseudo-disjunctive \TRIs-nB}\label{sse:disjunctive}

Up to this point, we have investigated the properties of Model $\M$
which is closely linked to \nTBpc and some of its variants.
We now briefly examine \nTBpd (defined in Section~\ref{sssec:etrinB}).

Let us first observe, as in Remark~\ref{rem:ETRI1et3monotone}, that \nTBpd and \nTBIpd respect the dominance relation. Therefore, by Proposition~\ref{prop:partRespectsDom}, we have that these models satisfy linearity. Hence all partitions determined either by \nTBpd or by \nTBIpd have a representation in Model $\M$.

Whether or not these models are equivalent to Model $\M$ is still unclear. The interested reader may refer to \citet[][Section~5]{bouyssouMP2020arXivETRInB}, for a theoretical look at \nTBpd. This reference contains examples showing that the study of this model is more complex than that of the pseudo-conjunctive version. These examples suggest that \nTBpd might be strictly included in Model $\M$, but this question remains open.

The extra complexity involved in studying the pseudo-disjunctive model was already pointed out by \citet{BouyssouMarchantEJOR2015} in the case of \TBpd. Their analysis concludes that this is mainly due to the fact that \TBpc and \TBpd are not dual of each other \citep[see][Section~5.3, for more detail]{bouyssouMP2020arXivETRInB}.

\section{Discussion}\label{se:discussion}\label{sec:discussion}


\subsection{Summary}\label{ssec:summaryDisc}
Using classical tools from conjoint measurement, we have proposed a
new interpretation of the Decomposable model \eqref{def:model:D} introduced by Goldstein, \ie of linear twofold partitions of a finite product set $X = \prod_{i=1}^{n} X_i$. Any such partition can be represented in Model $\M$ using an appropriate set of limiting profiles. It can also always be represented in Model $(\Ms^{u})$, the unanimous model, using as limiting profiles the set of minimally acceptable elements in $X$
. Some linear twofold partitions have a unique representation in Model $\M$. It is then a unanimous one. Some other linear twofold partitions have several  representations in model $\M$, one of which is unanimous.

In case $X_i$ is a finite subset endowed with a linear ordering $\geq_i$ (\eg $X_i$ is the range of the real-valued evaluation function $g_i$ \wrt criterion $i$) and the partition $\PART$ of $X=\prod_{i=1}^{n}X_i$ respects dominance, then it is linear \wrt the weak order $\succsim_i$ on $X_i$ that is compatible with the linear order $\geq_i$ on $X_i$. Any such twofold partition is representable in Model \nTBpc using an appropriate set of limiting profiles and other model parameters (thresholds, additive criteria weights). It can also be represented in model \nTBIpc (again with appropriate, possibly different, limiting profiles and parameters). Of course, it can also be represented in Models \eqref{def:model:D} or in Model $\M$ (possibly using a set of winning coalitions $\F$ that cannot be represented by additive weights and a threshold). Since all these models contain as a particular case the unanimous model $(\Ms^{u})$, it may happen that the sole representation of the twofold partition in these models is the unanimous one. In some cases, but not always, there exists a more  synthetic representation with fewer limiting profiles and a nontrivial set of winning coalitions.

Note that any sorting model producing partitions respecting dominance and able to determine any partition generated by the unanimous model is equivalent to Model \eqref{def:model:D}. Not all sorting models respecting dominance however are able to determine any partition produced by the unanimous model. For example, the additive model \eqref{def:model:Add} is strictly included in Model \eqref{def:model:D}.

Somewhat surprisingly, while Bernard Roy had
always championed outranking approaches
as an alternative to the classical additive
value function model,
it turns out that the last \textsc{Electre} method that he published before he passed away, \lnTB, contains the additive value function model for sorting as a particular case. We think that this unexpected conclusion is a plea
for the development of axiomatic studies in the field of decision with multiple attributes, as already advocated in \citet{Manifesto}, more than 25 years ago.

\subsection{Perspectives for elicitation and learning}\label{ssec:learning}
In this section, we focus on situations in which a model of assignment respecting dominance has to be learned solely on the basis of assignment examples. We thus assume that either we do not have the opportunity to interact with a decision maker or, if such a possibility exists, we decided to only ask her questions in terms of assignments to the classes of the partition. In such a perspective, the intuitive content of the underlying models plays no role, while mastering the complexity of these models is important as observed in  Section \ref{sssec:observations}, item 3.

\nTBpc is equivalent to Model \M, which is equivalent to the Decomposable model \eqref{def:model:D} and thus to the model based on ``at least'' decision rules (see Section~\ref{ssubsec:D1:int}). Techniques have been proposed to learn a decision rule model
\citetext{see \citealp{GrecoMatarazzoSlowinski99bEJOR,GMS2001EJOR,GMS2001,GMS02Control,GMS2016},
\citealp[and, for a recent application of these techniques,][]{AbastanteJMCDA}}. A large scale recent application to the detection of frauds in car loans applications is described in \citet{BLASZCZYNSKIetalExpertSystAppl2021}. The dataset involves 26\,187 loans among which 405 were fraudulently obtained. The dominance-based rough set approach (DRSA) is applied. It determines ``at least'' decision rules, which approximately reproduce the partition in fraudulent and non-fraudulent loans. The approach outperforms two classical machine learning techniques (random forest and support vector machine).
Methods such as DRSA can directly be used to learn an \nTBpc model since a limiting profile can immediately be deduced from any ``at least'' decision rule $d$. Indeed, the latter specifies minimal levels to be attained on a subset $N_d \subseteq N$ of criteria in order to be assigned to category $\Aa$. A minimally acceptable element associated to such a rule is the $n$-tuple whose components corresponding to the criteria in $N_d$ are set to the minimal levels specified in the rule. The components corresponding to any criterion $i \not\in N_d$ are set to the minimal element in $X_i$ (\wrt the linear order $\geq_i$). An alternative satisfies rule $d$ iff it is at least as good as the minimally acceptable element associated to the rule. This correspondence between rules and minimally acceptable elements provides a description of the partition in the unanimous model $E^{u}$.

Because Model \eqref{def:model:D} is quite general and the learning sets of assignment examples usually of limited size, using these techniques
is not entirely straightforward and, \eg may lead, in the decision rule model,
to a large number of rules. Moreover, these techniques, when used for learning an \nTBpc model, produce indeed an \nTBpc model but under the form of a unanimous model, \ie in terms of a set of minimally acceptable alternatives, not under a more compact form, even when there is one.

Having at hand alternative descriptions of Model \eqref{def:model:D} may offer an opportunity to control the complexity of the learned models.
It is therefore important to investigate particular
cases of Model $\M$, in which the cardinality
of the set of profiles $\Pro$ is restricted. Unfortunately, the problem seems to be difficult.
This is the subject of a companion paper that deals with these more technical issues
\citep{BouyssouMarchantPirlot2019ETRInprof}. This companion paper only
analyzes the particular case
of two profiles coupled with unanimity. Even in this apparently simple case,
the problem is not easy.
Hence, our analysis also leaves open the study of the gain of expressiveness brought
by increasing the size of the set of profiles.
Going from a single profile, the case studied in \citet{BouyssouMarchant2007:I},
to an arbitrarily large number of profiles, the case
implicitly studied in Section~\ref{sec:model:m},
leads to a \emph{huge} gain in expressiveness.
Is this gain already present when going from
a single profile to a small number of profiles?
This question is clearly important as a guide to learning the parameters of \nTB.
Our analysis of the case of two profiles coupled with unanimity
shows that it is unlikely that a purely axiomatic investigation
will allow us to obtain clear answers to this question.
Hence, this is also a plea to combine axiomatic work with other
types of work, \eg based on computer simulations.

Instead of constraining the number of limiting profiles in the unanimous model, an alternative approach would consist of exploring \nTBpc or \nTBIpc with  restricted number of limiting profiles\footnote{With such models, even with one limiting profile, the number of minimally acceptable alternatives (\ie the number of limiting profiles in the equivalent unanimous model) can be large. In the simplest case of \nTBIpc with one limiting profile and no veto, a minimally acceptable alternative takes the values of the limiting profile minus the indifference threshold $qt_i$ on a minimal winning coalition of criteria and the minimal value in $X_i$ on all other criteria. If the model is such that a coalition is winning whenever it contains at least $n/2$ criteria (this is obtained by setting all criteria weights to $1/n$ and the cutting level $\lambda$ to $1/2$), then the number of minimal winning coalitions is maximal and is equal to the Sperner number $n\choose \lceil n/2\rceil$ \citep[see, \eg ][pp.~116-118]{CaspardBookEng2012}.  Therefore, for such a model, the maximal number of minimally acceptable alternatives is equal to this number.}. Models \nTBIpc and  $\M$ with one limiting profile are well-known. Model $\M$ with one limiting profile is the noncompensatory sorting (NCS) model characterized by \citet{BouyssouMarchant2007:I}. The particular case in which the winning coalitions can be represented by additive weights and a majority threshold corresponds to model \nTBIpc with one limiting profile. In the absence of veto, this model is known as MR-Sort \citep{LeroyMousseauPirlot2011,Sobrie2018}

Several methods have been proposed for learning Model $\M$ with one limiting profile \citep{Belahcene2018COR} or its particular case Model \nTBIpc with one limiting profile \citep{LeroyMousseauPirlot2011,Sobrie2018,SobriePopulation2017}. These methods rely on various techniques such as mixed integer programming (MIP), Boolean satisfiability algorithms (SAT, MaxSAT) or metaheuristics. The size of sets of assignment examples that exact methods (such as MIP, SAT or MaxSAT) can deal with is limited. In contrast, the metaheuristic designed by \citet{SobrieADT13,Sobrie2018} or that by \citet{OlteanuMeyer2014} competes with state-of-the-art machine learning algorithms on real datasets \citep{TehraniCDH12,TehraniIEEEFuzzy12}. The size of these datasets, which are benchmarks commonly used in machine learning, ranges from 120 to 1728 alternatives, the number of attributes, from 3 to 8 and the number of categories, from 2 to 9.

Characterizing Model $\M$ with a fixed small number (\eg 2 or 3) of limiting profiles seems very difficult. The only thing that can easily be provided is an upper bound on the number of equivalence classes of the relations $\succsim_i$ induced by the corresponding twofold partition on the scale $X_i$ of each criterion (see Observation~2 in Section~\ref{sssec:observations}, for an illustration) and correlatively, an upper bound on the maximal number of minimally acceptable alternatives (see footnote \ref{foot:TOP}, p.~\pageref{foot:TOP}). Extending the approaches referred to above for learning models with more than one limiting profile has not been done yet and does not seem straightforward either. In case a formulation for learning such models when the number of limiting profiles is limited to 2 or 3 would prove operational, then an incremental learning approach could be envisaged. Start with fitting Model $\M$ with one limiting profile to the data. If assignment accuracy is not satisfactory, proceed with fitting a model with two profiles, etc.

Turning to the learning of an \nTBpc model (using an \textsc{Electre}~III outranking relation), observe first that the case with one limiting profile corresponds to the classical \TRI-pc model. Much effort has been devoted to develop learning methods for this model \citep[\eg][]{MousseauSlowinski1998jogo,NgoTheMousseau2002,DoumposMMZ09}. The genetic algorithm proposed by \citet{DoumposMMZ09} has been tested on a real dataset (in the banking sector) involving 100 alternatives evaluated on 7 criteria and assigned to 3 categories. It has also been tested on artificial datasets, involving up to 1000 alternatives in the learning set, assigned to categories using randomly generated \TRIs models. No exact methods have been developed to date to learn \nTBpc.  The genetic algorithm proposed by \citet{FernandezSoftComputing2019} for learning \nTBpc has shown good performance on a real case study involving 81 assignment examples (R \& D projects) evaluated on 4 criteria assigned to 8 categories. It has also been
tested on artificial datasets assigned to categories by randomly generated \nTBpc (using 5 limiting profiles in each category boundary).

In conclusion, the general approach, based on the decision rule model and techniques, without restrictions on the generality of Model \eqref{def:model:D}, is available but does not allow to easily control the simplicity of the learned model \citep[see][]{GrecoMSS00a,DembczynskiPS03a}.  Current methods, such as the genetic algorithm proposed by \citet{FernandezSoftComputing2019} are usable. Unfortunately, the scalability of these algorithms is difficult to assess since they have not been tested on common benchmark datasets
\footnote{Incidentally, we came across the recent paper by \citet{Silva2021} in which the decision rule approach (DRSA) is applied to the rating of sovereign risk; the results are then compared with those obtained using an additive model and MR-Sort. Unfortunately, the dataset involved is a small one (36 countries).}
.

Alternative approaches remain to be developed. Ideally, they should fulfill the following three requirements.
\begin{enumerate}
  \item Focus on a well-defined, preferably characterized, family of assignment models forming a proper subset of all assignment models \eqref{def:model:D} respecting the linearity property or the dominance relation.
  \item The models in this family should have a compact description, \ie there should be a synthetic, interpretable, manner of describing the set of minimally acceptable alternatives.
  \item Learning these models should be computationally tractable, \ie there should be an algorithm able to fit, in reasonable computing time, a model in the family to a set of assignment examples involving several hundreds up to a few thousands assignments.
\end{enumerate}

Model $\M$ and \nTBpc with a restricted number of limiting profiles fulfill the first requirement and the second but not the third (except perhaps for Model $\M$ with one limiting profile). The additive model \eqref{def:model:Add} is a candidate that checks all three boxes. It is closely related with the UTADIS technique and its variants (as already mentioned in Section \ref{subsec:Goldstein}). However, its interpretation is quite at a distance from that of \textsc{Electre} based models, which rely on the idea of one or several limiting profiles and outranking relations. These are interpreted as expressing requirements on each criterion that an alternative should ideally fulfill in order to be acceptable.
A challenging research issue is thus to define a family of models %
within the Decomposable model \eqref{def:model:D} that are in the spirit of the \textsc{Electre} methods and fulfills the above three requirements.

\subsection{Future research and work in progress}\label{ssec:researchIssues}

Theorem~\ref{th:main:nB} is a simple result that establishes the equivalence of the Decomposable model \eqref{def:model:D} with \nTBpc and related models in the spirit of \TRI.
Besides issues related to learning, this result leaves open a number of interesting problems that we intend to deal with in later studies.
Among them, let us mention the following sets of questions. 

\emph{Algorithmic questions}.  Is it easy to test whether a partition $\PART$ satisfies linearity?
Are there efficient algorithms to
find a linear partition close to a partition that is not linear?
Is it easy to test whether it is possible to build a linear partition on the basis of partial information about $\Aa$ and $\Uu$?
Similar questions arise, when there is a supplementary constraint on the size of the set of profiles.


\emph{Combinatorial questions}.  Given a set $X = \prod_{i=1}^{n} X_{i}$, can we
  devise (easy to evaluate) formulae for the maximal
  number of objects in $\As$ (this is related to the size of the largest
  antichain\,\footnotemark{} in a direct product of
  chains,
  \citealp[see][for the case of the direct product of chains of the same length]{Sander1993})?
  What is the number of twofold partitions of $X = \prod_{i=1}^{n} X_{i}$
  that can be represented in Model $\M$
    \citep[this is related to the famous problem of Dedekind numbers, see][]{Kahn2002,Kisielewicz1988,Uyanik2017}?

%
\footnotetext{In \citet{BouyssouMarchantPirlotAntichains}, we give the proof
of a result that has already appeared
in the grey literature but for which no proof was available. It states that
if the chain on $\Seti$ has $m_{i}$ elements, the maximal size of an antichain
in $\Set = \prod_{i=1}^{n} \Seti$ partially ordered by $\qodom$ is
\begin{equation*}
\sum_{I \subseteq \N: m_{I} < h - n} \binom{h - m_{I} -1}{n-1} (-1)^{\card{I}},
\end{equation*}
where $m_{I} = \sum_{i \in I} m_{i}$ and
\begin{equation*}
h = \left\lfloor \frac{n + \sum_{i \in \N} m_{i}}{2} \right\rfloor.
\end{equation*}
The reader will check that this number grows fast with the vector $(m_{i})_{i\in\N}$. \label{foot:TOP}}

\emph{Questions linked to the number of profiles}.
Given a learning set of assignment examples that is compatible with Model \eqref{def:model:D}, what is the minimal number of limiting profiles of a unanimous model (or of a Model $\M$ or of an \nTBpc model) that exactly restores the assignments?

%



\addcontentsline{toc}{section}{References}
\bibliographystyle{plainnat}
\bibliography{compens}
\newpage\pagenumbering{roman}
\appendix
\part*{Appendix: To appear as supplementary material}
\addcontentsline{toc}{part}{Appendix: Supplementary material}


\section{Binary relations}\label{app:binary}

We use a standard vocabulary for binary relations.
For the convenience of the reader and in order to avoid any misunderstanding,
we detail our vocabulary here.
A binary relation $\rel$ on a set $Z$ is a subset of $Z \times Z$.
For $x, y \in Z$, as is usual, we will often write
$x \rel y$ instead of  $(x,y) \in {\rel}$.

Let $\rel$ be a binary relation on $Z$. We define:
\begin{itemize}
\item the asymmetric part $\arel$ of $\rel$ as $x \arel y
      \iff \left[x \rel y \text{ and } \Not{y \rel x} \right]$,
\item the symmetric part $\irel$ of $\rel$ as $x \irel y
      \iff \left[x \rel y \text{ and } y \rel x\right]$,
\item the symmetric complement $\jrel$ of $\rel$ as $x \jrel y
      \iff \left[\Not{x \rel y} \text{ and } \Not{y \rel x}\right]$,
\end{itemize}
for all $x, y \in Z$.

A binary relation $\rel$ on $Z$ is said to be:
\begin{enumerate}[(i)]
\item \emph{reflexive} if $x \rel x$,
\item \emph{irreflexive} if $\Not{x \rel x}$,
\item \emph{complete} if $x \rel y$ or $y \rel x$,
\item \emph{symmetric} if $x \rel y$ implies $y \rel x$,
\item \emph{asymmetric} if $x \rel y$ implies $\Not{y \rel x}$,
\item \emph{antisymmetric} if $[x \rel y$ and $y \rel x]$ $\Rightarrow$ $x = y$,
\item \emph{transitive} if $[x \rel y$ and $y \rel z]$ $\Rightarrow$ $x \rel z$,
\item \emph{Ferrers} if $[x \rel y \text{ and } z \rel w] \Rightarrow
[x \rel w \text{ or } z \rel y]$,
\item \emph{semitransitive} if $[x \rel y \text{ and } y \rel z] \Rightarrow
[x \rel w \text{ or } w \rel z]$,
\end{enumerate}
for all $x,y,z,w \in Z$.

We list below a number of remarkable structures. A binary relation
$\rel$ on $Z$ is said to be:
\begin{enumerate}[(i)]
\item a \emph{weak order} (or \emph{complete preorder}) if it is complete and transitive,
\item a \emph{linear order} if it is an antisymmetric weak order,
\item a \emph{semiorder} if it is reflexive, Ferrers and semitransitive,
\item a \emph{strict semiorder} if it is irreflexive, Ferrers and semitransitive,
\item an \emph{equivalence} if it is reflexive, symmetric, and transitive,
\item a \emph{partial order} if it is reflexive, antisymmetric and transitive.
\end{enumerate}
Notice that a reflexive and Ferrers relation must be complete. Similarly
an irreflexive and Ferrers relation must be asymmetric.

When $\rel$ is an equivalence relation on $Z$,
 the set of equivalence classes of $\rel$ on $Z$ is denoted $\quotient{Z}{\rel}$.
A \emph{partition} of $Z$ is a collection of nonempty subsets of $Z$
that are pairwise disjoint
and such that the union of the elements in this collection is $Z$.
It is clear that, when $\rel$ is an equivalence relation on $Z$,
$\quotient{Z}{\rel}$
is a partition of $Z$.


When $\rel$ on $Z$ is a semiorder, its asymmetric part $\arel$ is irreflexive,
Ferrers and semitransitive, \ie a strict semiorder.

Any Ferrers and semitransitive
$\rel$ on $Z$ (which includes semiorders and strict semiorders) induces a
weak order $\rel^{wo}$ on $Z$ that is defined as follows:
\begin{equation}\label{eq:wo:so}
a \rel^{wo} b \text{ if }
\forall c \in Z, [b \rel c  \Rightarrow  a \rel c]
\text{ and }
[c \rel a  \Rightarrow  c \rel b].
\end{equation}
If $\rel$ is a semiorder and $\brel$ is its asymmetric part,
it follows that
${\rel^{wo}} = {\brel^{wo}}$. The weak order induced by a semiorder
is identical to the one induced by its asymmetric part.

Let $\rel$ and $\brel$ be two semiorders on $Z$ such that
${\rel} \subseteq {\brel}$.
We say that $(\rel, \brel)$ is a \emph{nested chain} of semiorders.
Let $\rel^{wo}$ (\resp $\brel^{wo}$) be the weak order on $Z$ induced by
$\rel$ (resp.\ $\brel$).
If a nested chain of semiorders ${\rel} \subseteq {\brel}$ is such that the relation
${\rel^{wo}} \cap {\brel^{wo}}$ is complete (and therefore is a weak order),
we say that the nested chain of semiorders $(\rel, \brel)$ is \emph{homogeneous}
\citep[][]{DoignonMonjardetRoubensVincke88JMP}.

Finally, let us note that $\rel$ is a semiorder on a finite set $Z$ iff
there are a real-valued function
$f$ on $Z$ and a positive number $s > 0$ such that, for all $a, b \in Z$,
$a \rel b \Leftrightarrow f(a) \geq f(b) - s$

\section{Sketch of the proof of Proposition~\ref{prop:minimal}}\label{app:maximal}

Suppose that $\Max(\rel, B)$ is empty.
Let $x \in B$.
By hypothesis, $x$ does not belong to $\Max(\rel, B)$.
This implies that there is
$w_{1} \in B$ such that $w_{1} \arel x$.
Clearly, this implies that $w_{1}$ is distinct from $x$, because $\arel$ is irreflexive.
But $w_{1}$ does not belong to $\Max(\rel, B)$.
This implies that there is
$w_{2} \in B$ such that $w_{2} \arel w_{1}$.
Clearly, this implies that $w_{2}$ is distinct from both $w_{1}$ and $x$.
Continuing the reasoning leads to postulating the existence of a chain of elements
$w_{i}$, $i \in \Natp$, that are all distinct (otherwise, the transitivity of $\arel$
will lead to violate irreflexivity). This violates
the finiteness of $B$. Hence, $\Max(\rel, B)$ must be nonempty.
The proof that  $\Min(\rel, B)$ must be nonempty is similar.
The fact that, for all $x,y \in \Max(\rel, B)$,
we have $\Not{x \arel y}$ is clear from the definition of $\Max(\rel, B)$.
The same is clearly true with $\Min(\rel, B)$.

Suppose now that $x \in B$ and there is no
$y \in \Max(\rel, B)$ such that $y \rel x$.
If $x \in \Max(\rel, B)$, the contradiction is established, because $\rel$ is reflexive.
Suppose, hence, that $x \notin \Max(\rel, B)$.
There is $w_{1} \in B$ such that $w_{1} \arel x$.
But it is impossible that $w_{1}$ belongs to $\Max(\rel, B)$.
This implies that there is $w_{2} \in B$
such that $w_{2} \arel w_{1}$. Because,
$\arel$ is transitive, it is impossible that $w_{2} \in \Max(\rel, B)$.
Because $\arel$ is asymmetric are transitive, it is impossible that
$w_{2}$ is identical to $w_{1}$ or to $x$.
Continuing the same reasoning,
leads to postulating the existence of a chain of elements
$w_{i}$, $i \in \Natp$, that are all distinct.
This violates the finiteness of $B$.
Hence, there exists $y \in \Max(\rel, B)$ such that $y \rel x$.
The proof that if $x \in B$, there is $z \in \Min(\rel, B)$ such that
$x \rel z$ is similar. \hfill  $\Box$

\section{Example: Minimally \acceptable alternatives for rational evaluations with one decimal digit}\label{app:sec:ExampleDecimal}
The \nTB model specified in Section \ref{sse:exampleEtrinB} uses two profiles
$p^1 = (8,7,5)$ and $p^2=(5,6,8)$. The indifference, preference and veto thresholds are, respectively, $qt_i=1$, $pt_i=2$ and $vt_i=4$, the same for all criteria $i=1,2,3$. All criteria have the same weight $w_i=\frac{1}{3}$ and the cutting threshold $\lambda = .6$.
We apply this model to the set $X$ of alternatives whose evaluations are rational numbers with one decimal digit ranging in $[0,10]$. The set of minimally \acceptable alternatives is determined below.

For $x\in X$ to be \acceptable, $c(x,p^1)$ or $c(x,p^2)$ has to be at least equal to $\lambda = .6$. We develop the consequences of this condition for $p^1$, the case of $p^2$ being similar. This condition entails that $x_i$ must be strictly greater than $p^1_i - 2$ for at least two criteria. We distinguish two cases:  Case~1: $x_i$ is strictly greater than $p^1_i - 2$ on all three criteria;  Case~2: $x_i$ is strictly greater than $p^1_i - 2$ on exactly two criteria and less than this value on the third criterion.

\subsection{Case 1}\label{sse:Ex3crit}
Let $c_i$ be shorthand for $c_i(x,p^1)$, $i=1,2,3$. If $x \in \Aa$, $c(x,p^1) = \sum_{i=1}^{3}w_i c_i= 1/3 \sum_{i=1}^{3} c_i \geq .6$. If $x_i$ is strictly greater than $p^1_i - 2$ for all $i$, then we have $\sum_{i=1}^{3}c_i \geq 1.8$ with $c_i >0$ for all $i$. The alternative $y= (7,6,4)$ is the minimal one realizing $c(y,p^1)=1$. With respect to $(7,6,4)$, we may decrease all coordinates by a total of $1.2$ while remaining in $\Aa$.  For instance, for $x=(6.5, 5.8, 3.5)$, we have $\sum_{i=1}^{3} c_i= 1.8$ and $x$ is minimally \acceptable. There are actually $11 \choose 2$ =  55 ordered partitions of 12 objects (12 tenths) in three nonempty subsets. Among them, 3 partitions have a class of cardinal 10, which we must exclude. Hence, there are 52 ways of decreasing each coordinate of $(7,6,4)$ by at least one tenth, for a total amount of $12/10$ while
yielding rational coordinates with one decimal digit, that are respectively strictly greater  than $6,5,3$.  To this we have to add the different ways of decreasing two of the coordinates of $(7,6,4)$ by a total amount of $12/10$, while keeping unchanged the value of the third coordinate. There are $3 \times 7 = 21$ such alternatives. Hence there are $52 + 21 = 73$ minimally \acceptable elements of this type for $p^1$ and 73 for $p^2$.

\subsection{Case 2}\label{sse:Ex2crit}
The second type of minimally \acceptable elements $x$ satisfies $x_i > p^1_i -2$ for two values of $i$; the other coordinate does not satisfy this inequality. Assume that the latter coordinate is $i=3$. The condition $c(x,p^1) \geq .6$ is only satisfied in the following 6 cases:
\begin{enumerate}
  \item $(x_1,x_2) = (7,6)$ and $x_3 < 3$; in such a case $c(x,p^1)= 2/3$;
  \item $(x_1,x_2) = (6.9,6)$ or $(7,6.9)$ and $x_3 < 3$; in such a case $c(x,p^1)= 19/30$;
  \item $(x_1,x_2) = (6.8,6)$ or $(6.9,5.9)$ or $(7,5.8)$ and $x_3 < 3$; in such a case $c(x,p^1)= 6/10$;
\end{enumerate}
Let us now compute in each case, the minimal value of $x_3$ such that $\sigma(x,p^1) \geq .6$.
\begin{enumerate}
  \item If $x_3 = 1.6$, $d_3(1.6, 5) = 0.7 > 2/3 = c(x,p^1)$. We have $\frac{1-d_3(1.6,5)}{1-c(x,p^1} = \frac{3/10}{1/3}= 9/10$. Therefore $\sigma((7,6,1.6),p^1)= 2/3 \times 9/10 = .6$. Taking $x_3 <1.6$ would lead to an \unacceptable alternative. Hence $(7,6,1.6)$ is minimal in $\Aa$.
  \item If $x_3 = 1.7$, $d_3(1.7, 5) = 0.65 > 19/30 = c(x,p^1)$. We have $\frac{1-d_3(1.7,5)}{1-c(x,p^1} = \frac{7/20}{11/30}= 21/22$. Therefore, for $(x_1,x_2) = (6.9,6)$ or $(7,6.9)$,  $\sigma((x_1,x_2,1.7),p^1)=  19/30 \times 21/22 \approx .6045 >.6$. Taking $x_3 <1.7$ would lead to \unacceptable alternatives. Hence $(6.9,6,1.7)$ and $(7,5.9,1.7)$ are minimal in $\Aa$.
  \item If $x_3 = 1.8$, $d_3(1.8, 5) = 0.6  = c(x,p^1)$. We have $\sigma((x_1, x_2, 1.8),p^1)= c((x_1, x_2, 1.8),p^1) = .6$, for $(x_1,x_2) = (6.8,6)$ or $(6.9,5.9)$ or $(7,5.8)$. Taking $x_3 <1.8$ would lead to \unacceptable alternatives. Hence $(6.8,6,1.8)$, $(6.9, 5.9, 1.8)$ and $(7,5.8,1.8)$ are minimal in $\Aa$.
\end{enumerate}
There are thus 6 minimally \acceptable alternatives with their third coordinate smaller than $p^1_3 - 3$. By symmetry, there are 6 minimally \acceptable alternatives having a third coordinate smaller than $p^1_i - 3$. Therefore, there are 18 minimally \acceptable alternatives of the second type for $p^1$ and similarly for $p^2$.

Summing up, the total number of minimally \acceptable alternatives is $2 \times (73+18) = 182$. None of these dominates another, as it can be easily verified.
\section{Example: A linear partition having only a unanimous representation in model $\M$} \label{app:OnlyReprIsUnanim}
%
The example has $n= 4$ and $X_{1} = X_{2} = X_{3} = X_{4} =  \{0, 1, 2\}$.
We let $\Aa = \{x \in X: \sum_{i=1}^{4} x_{i} \geq 6\}$. There are
$3^{4} = 81$ objects in $X$, $15$ are in $\Aa$, while $66$ are in $\Uu$.

Observe first that, on all attributes, we have
$2 \asi 1 \asi 0$. Indeed, with $i=1$, we have:
\begin{equation*}
\begin{aligned}
(2, 0, 2, 2) \in \Aa, && (1, 0, 2, 2)  \in \Uu,\\
(1, 1, 2, 2)  \in \Aa, && (0, 1, 2, 2)  \in \Uu.
\end{aligned}
\end{equation*}
The same relations clearly hold on all attributes since the problem is symmetric.

This partition clearly has a representation in Model $\M$ with $\F = \{N\}$
and
a set of profiles
consisting of all $10$ objects in the class
$6$ (\ie having a sum of components equal to $6$). By construction, these
$10$ profiles are not linked by dominance (this is a representation in model $(\Ms^{c})$).

Our objective is to try obtaining a representation in Model $\M$ using a set
$\F$ that is not reduced to $\{\N\}$.
Notice first that bringing the veto relations into play will not help us do so.
Indeed, it is easy to check
that if a representation exists in model $\M$, a representation exists in
Model $(\Ms^{c})$ (because whatever $x_{i}$, we can find $a_{-i}$ such that
$(x_{i}, a_{-i}) \in \Aa$). Hence, let us try to find a representation in Model $(\Ms^{c})$.

This clearly excludes to take any object in the class $6$
as a profile. Indeed, a family
$\F$ that is not reduced to $\{N\}$ would then
imply that some object in a class strictly lower than $6$
belongs to $\Aa$, which is false.
Hence, we must take as profiles objects belonging to the class $7$ or $8$.

Because profiles cannot dominate one another, if we take the object $(2, 2, 2, 2)$
as a profile, it must be the only one.
We know that
$(2, 2, 1, 1) \in \Aa$.
Hence, we must have $\{1, 2\} \in \F$.
This is contradictory. Indeed, since $\{1, 2\} \in \F$, we should have
$(2, 2, 0, 0) \in \Aa$, a contradiction.

Hence the set of profiles must consist exclusively of objects belonging to the class $7$.

Suppose that there is a unique profile, \eg
$(2, 2, 2, 1)$.
It is clear that the set $\{1,2,3\}$ must be included in all elements of $\F$ (otherwise we would have
an object in the class $5$ belonging to $\Aa$).
Because $(2, 2, 2, 0) \in \Aa$, it must be true that $\{1,2,3\}$ is an element of $\F$,
which must  therefore be equal to
$\{\{1, 2, 3\},\{1, 2, 3, 4\}\}$.
This is contradictory since we know that
$(0, 2, 2, 2) \in \Aa$. It is easy to see that, the problem being symmetric,
it is therefore impossible to have a representation using a single profile
from the class $7$.

A similar reasoning can be made if we consider the cases of two or three profiles
from the class $7$ as profiles.

%

Suppose finally that we choose all four profiles from the class $7$:
$(1, 2, 2, 2)$, $(2, 1, 2, 2)$, $(2, 2, 1, 2)$, and $(2, 2, 2, 1)$.
Using the same reasoning as above, the set $\F$ must contain the sets
$\{2,3,4\}$, $\{1,3,4\}$, $\{1,2,4\}$, and $\{1,2,3\}$, since
$(0, 2, 2, 2)$, $(2, 0, 2, 2)$, $(2, 2, 0, 2)$ and
$(2, 2, 2, 0)$ are all in $\Aa$.
But this is contradictory since this would imply that
$(0, 1, 2, 2) \in \Aa$ (since $(2, 1, 2, 2)$ is a profile and $\{2,3,4\} \in \F$).

Therefore, the only possible representation of this partition in Model $\M$ must use as profiles
all 10 elements in the class $6$ together with $\F = \{N\}$.
\section{An example of linear partition not representable in Model \eqref{def:model:Add}}\label{app:AddNotEqDecomp}

Let $X = \prod_{i=1}^4 X_i $, where $X_i=\{0,1\}$. Let $\Aa= \{1100, 0011, 1110, 1101, 1011, 0111, 1111\}$ and $\Uu$ the complement of $\Aa$ in $X$. The partition $\PART$ respects the dominance relation determined by the natural order on $X_i$, for all $i$. This partition cannot be represented in Model \eqref{def:model:Add}. Assuming the contrary would entail the following:
\begin{align*}
  u_1(1) + u_2(1) + u_3(0) + u_4(0) & >0 \\
  u_1(0) + u_2(0) + u_3(1) + u_4(1) & >0.
\end{align*}
This implies that $\sum_{i=1}^{4} u_i(1) + \sum_{i=1}^{4} u_i(0) > 0$. Since $1010$ and $0101$ belong yo $\Uu$, we should also a have:
\begin{align*}
  u_1(1) + u_2(0) + u_3(1) + u_4(0) & \leq 0 \\
  u_1(0) + u_2(1) + u_3(0) + u_4(1) & \leq 0.
\end{align*}
Therefore we must have $\sum_{i=1}^{4} u_i(1) + \sum_{i=1}^{4} u_i(0) \leq 0$, a contradiction.
\end{document}